\documentclass[twocolumn,twoside,pra,superscriptaddress]{revtex4}
\usepackage[utf8]{inputenc}

\usepackage[pdftex]{graphicx,color}
\usepackage{amsmath}
\usepackage{amssymb}
\usepackage{epsfig}
\usepackage{amsthm}
\usepackage{tikz}
\usepackage{wrapfig}
\usepackage{comment}
\usepackage{latexsym,revsymb}

\newcommand{\ket}[1]{| #1 \rangle}
\newcommand{\bra}[1]{\langle #1 |}
\newcommand{\av}[1]{\langle #1 \rangle}
\newcommand{\de}{d_{\mathrm{eff}}}
\newcommand{\f}[2]{\textstyle{\frac{#1}{#2}}}
\newcommand{\tr}[1]{\mathrm{tr}\! \left[#1\right]}
\newcommand{\abs}[1]{\left| #1 \right|} 
\newcommand{\avg}[1]{\left< #1 \right>} 
\newcommand{\norm}[1]{\| #1 \|} 
\newcommand{\proj}[1]{| #1 \rangle\langle #1 |}
\newcommand{\braket}[2]{\langle #1 | #2 \rangle}
\newcommand{\vertiii}[1]{{\left\vert\kern-0.25ex\left\vert\kern-0.25ex\left\vert #1 
    \right\vert\kern-0.25ex\right\vert\kern-0.25ex\right\vert}}

\newcommand{\bl}[1]{{\color{black} #1}}
\newcommand{\re}[1]{{\color{black} #1}}

\makeatletter
\newtheorem*{rep@theorem}{\rep@title}
\newcommand{\newreptheorem}[2]{%
\newenvironment{rep#1}[1]{%
 \def\rep@title{#2 \ref{##1}}%
 \begin{rep@theorem}}%
 {\end{rep@theorem}}}
\makeatother

\newtheorem{theorem}{Theorem}
\newreptheorem{theorem}{Theorem}

\newtheorem{corollary}{Corollary}
\newreptheorem{corollary}{Corollary}
\newtheorem{lemma}[theorem]{Lemma}
\newreptheorem{lemma}{Lemma}

\begin{document}
\title{Equilibration of quantum gases}
\author{Terry Farrelly}
\email{farreltc@tcd.ie}
\affiliation{Institut f{\"u}r Theoretische Physik, Leibniz Universit{\"a}t, Appelstra{\ss}e 2, 30167 Hannover, Germany}

\begin{abstract}
Finding equilibration times is a major unsolved problem in physics with few analytical results.  Here we look at equilibration times for quantum gases of bosons and fermions \bl{in the regime of negligibly weak interactions}, a setting which not only includes paradigmatic systems such as gases confined to boxes, but also Luttinger liquids and the \bl{free superfluid} Hubbard model.  To do this, we focus on \bl{two classes of measurements:\ (i)} coarse-grained observables, such as the number of particles in a region of space\bl{, and (ii) few-mode measurements, such as phase correlators} and correlation functions.  We show that, in this setting, equilibration occurs quite generally despite the fact that the particles are not interacting.  Furthermore, \bl{for coarse-grained measurements} the timescale is \bl{generally} at most polynomial in \bl{the number of particles} $N$, which is much faster than previous general upper bounds, which were exponential in $N$.  \bl{For local measurements on lattice systems, the timescale is typically linear in the number of lattice sites.  In fact, for one dimensional lattices, the scaling is generally linear in the length of the lattice, which is optimal.}  Additionally, we look at \bl{a few specific examples}, one of which \bl{consists of} $N$ fermions initially confined on one side of a partition in a box.  The partition is removed and the fermions equilibrate extremely quickly in time $O(1/N)$.
\end{abstract}

\maketitle
\section{Introduction}
\bl{Over} the past \bl{few} decade\bl{s}, there has been a major push to understand statistical physics by applying tools from quantum information.  One particularly pressing problem is understanding equilibration.  From everyday experience, we know it to be universal, as anything from a hot cup of tea to a spinning top will relax to a steady state eventually.  See figure \ref{fig:equilibration}.  However, our understanding of why equilibration occurs and how long it takes remains incomplete.  Progress has been dramatically helped by recent advances in experiments \cite{EFG14,LGS15}, where mesoscopic quantum systems can now be controlled extremely well, providing better and better playgrounds to probe properties of many-body systems.
\begin{figure}[ht!]
 \resizebox{6.5cm}{!}{\includegraphics{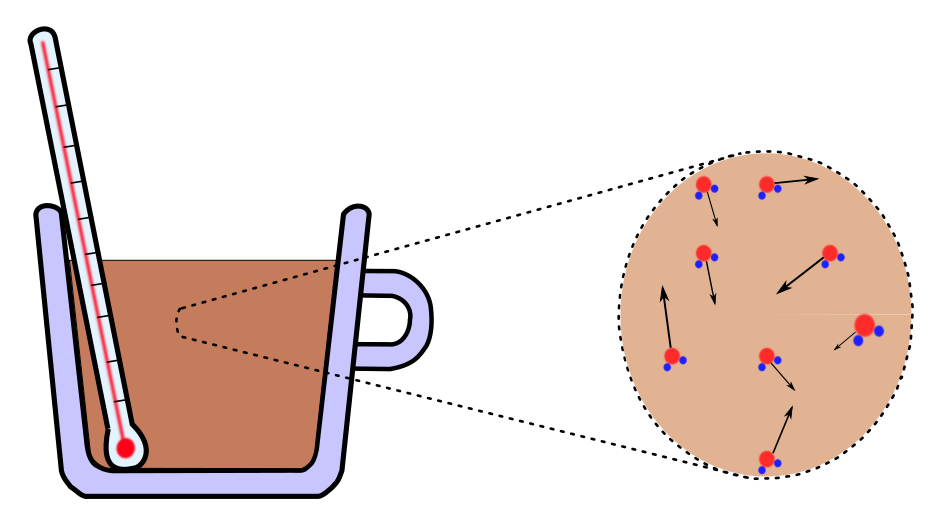}}
    \footnotesize{\caption[Equilibration]{Microscopically, a cup of tea is never in equilibrium:\ the molecules are constantly moving around, but we cannot measure this.  What we do measure is the temperature, according to which a hot cup of tea can reach equilibrium (room temperature).  This highlights an important point about equilibration, which is that it only occurs when we account for physical restrictions on what we can actually measure.\label{fig:equilibration}}}
\end{figure}

In \cite{Reimann08,LPSW09,Short10} it was proved that quantum systems will generally equilibrate with very weak assumptions on the Hamiltonian (which ensure, for one thing, that the system is not a collection of non-interacting subsystems).  But very little is known about the timescale.  This is crucial:\ if a system equilibrates but the timescale is the age of the universe, we will never actually observe it equilibrating in a lab.  Unfortunately, the best general upper bounds on the timescale \cite{SF12,RK12,MGLFS14} are far too large for even mesoscopic systems.  This is a consequence of the generality of the results.  Indeed, models were constructed effectively saturating these timescale bounds \cite{GHT13,MGLFS14}.

Imposing physical constraints on Hamiltonians and measurements has led to more realistic timescales in specific cases.  In fact, one of the earliest equilibration results was equilibration timescales for bosons evolving via the Hubbard Hamiltonian in the absence of interactions  \cite{CDEO08,CE10}.  More recently, equilibration timescales for small subsystems interacting with a large thermal bath were found \cite{GLMSW15}.  Along different lines, equilibration timescales were bounded by averaging over Hamiltonians, measurements or initial states \cite{UWE12,VZ12,Cramer12,BCHHKM12,MRA13,MGLFS14}.  For a review, see \cite{GE15}.

Here we will look at $N$ particle systems \bl{in the regime of negligible interactions} to see when equilibration occurs.  Experimentally, such situations appear often:\ Luttinger liquids \cite{RSLS15} are one example.

\bl{We look at two classes of measurement, which are natural for macroscopic and mesoscopic systems.  The first are coarse-grained measurements.}  These include the number of particles in some spatial region, the magnetization of fermions on a lattice, or the number of particles with different values of momentum.  The last of these arose in experiments with trapped Bose gases \cite{KWW06}, which\bl{,} in the limit of strong point-like interactions\bl{,} behave like free fermions.  \bl{The second type of measurements we consider are few-mode measurements.  Such measurements are crucial in many settings, and include correlation functions and phase correlators, which are important in ultracold atom experiments.}

\bl{First we will} look at some examples and then \bl{we will} show that equilibration of $N$ particle systems in this setting occurs quite generally and appears to be much faster than what general timescale bounds suggest.

\section{Equilibration}
Because there are recurrences for quantum systems with discrete spectra \cite{Bocchieri57,Schulman78}, the naive definition of equilibration as simply relaxation to a steady state is not sufficient.  Instead, we say a system equilibrates if it evolves towards a fixed state and stays close to it for most times.  To define what it means for two states to be close, we need a definition of distance between states.  For this to be realistic, we need to consider what measurements we can actually do.  For example, if we can do any measurement we want on a quantum system, then the distance between two states is best quantified by the trace distance, which allows us to calculate the maximum probability of distinguishing two states by doing a measurement \cite{NC00,Short10}.

In reality, for systems beyond a few qubits, there will be restrictions on the measurements we can do; for $10^{23}$ particles, clearly we are restricted to very coarse measurements.  With this in mind, a useful measure of distance is given by the distinguishability between states $\rho$ and $\sigma$, which is defined to be \cite{Short10}
\begin{equation}
 D_{\mathcal{M}}(\rho,\sigma)=\frac{1}{2}\max_{\{M_i\}\in\mathcal{M}}\sum_i\abs{\tr{\rho M_i}-\tr{\sigma M_i}},
\end{equation}
where $\mathcal{M}$ denotes the set of measurements we can do, and $\{M_i\}$ denotes a POVM measurement, with the positive operators $M_i$ satisfying $\sum_i M_i = \openone$.  \bl{POVM (Positive operator valued measure) measurements are more general than projective measurements.  This description may be necessary in situations where the measurement is not repeatable, for example.  Nevertheless, a POVM measurement is equivalent to a projective measurement on the system together with an ancilla \cite{NC00}.}

\bl{We denote} the infinite-time average of $\rho(t)$ \bl{by} $\avg{\rho}$.  \bl{I}f $D_{\mathcal{M}}(\rho(t),\avg{\rho})$ is small most of the time, \bl{then for all practical purposes $\rho(t)$ is indistinguishable from its time average $\av{\rho}$ most of the time.}  \bl{In that case,} equilibration has occurred.

Another notion of equilibration is equilibration of expectation values \cite{Reimann08}.  This works as follows.  Suppose we have the observable $M$ and we look at the quantity
\begin{equation}
 \Delta_M(t)=\frac{\abs{\tr{\rho(t) M}-\tr{\avg{\rho} M}}}{\|M\|},
\end{equation}
where $\|M\|$ is the operator norm of $M$.  This quantity tells us how close the expectation value of $M$ at time $t$ is to its time average, with the scale set by $\|M\|$.

If equilibration is to occur, we require that most of the time $\Delta_M(t)$ is smaller than some $\epsilon$, with $\epsilon$ chosen so that $\epsilon\|M\|$ is smaller than our experimental resolution.  

There is an important caveat here.  Even if expectation values equilibrate, we do not measure expectation values; we measure POVM outcomes.  In the examples we consider where equilibration of expectation values occurs, the fluctuations in measurement results are unobservably small.  This means that the measured value of $M$ is experimentally indistinguishable from $\tr{\rho(t) M}$ with extremely high probability.  Therefore, equilibration truly occurs.

\section{Gases of Bosons and Fermions}
The key step in getting estimates of the equilibration time for $N$ particle systems is equation (\ref{eq:18}) below, which will allow us to equate $\Delta_M(t)$ to the distinguishability for a single particle.

First, it will be useful to introduce some notation.  Let $\mathcal{H}$ be a single-particle Hilbert space, and let $\ket{i}$ denote an orthonormal basis.  Then we can define creation operators $a^{\dagger}_i$, acting on a fermionic Hilbert space, that create fermions corresponding to these states.  Equivalently, we may say $a^{\dagger}_i$ creates a fermion in mode $i$.  The fermionic Hilbert space is spanned by states with varying numbers of creation operators acting on $\ket{0}$, the empty state.  To avoid confusion, any state vectors written as kets are in the single-particle Hilbert space $\mathcal{H}$, with the exception of $\ket{0}$, which represents the empty state in a fermionic \bl{(or bosonic)} system.

The creation operator that creates a particle corresponding to the single-particle state $\ket{\psi}=\sum_i c_i \ket{i}$ is $a^{\dagger}(\ket{\psi})=\sum_i c_ia^{\dagger}_i$.  \bl{Suppose we have a single-particle Hamiltonian with discrete spectrum,
\begin{equation}
 H=\sum_{E}E \ket{E}\bra{E},
\end{equation}
where $E$ labels the energies.  There is a corresponding fermionic Hamiltonian, given by
\begin{equation}
 H_f=\sum_{E}E\, a^{\dagger}(\ket{E})a(\ket{E}).
\end{equation}
For any single-particle state $\ket{\psi}$, we also have
\begin{equation}
 e^{-iH_ft}a^{\dagger}(\ket{\psi})e^{iH_f t}=a^{\dagger}(e^{-iHt}\ket{\psi})=a^{\dagger}(\ket{\psi(t)}).
\end{equation}
The situation for bosons is similar.  The only difference is that, while fermionic creation and annihilation operators obey the canonical anti-commutation relations, bosonic creation and annihilation operators obey the canonical commutation relations.}

This is the basic idea behind second quantization, which allows one to take a single-particle system and upgrade it to a multi particle system \cite{BR97}.  \bl{Our goal here is to go in the opposite direction and to study equilibration of many-particle systems by moving to the single-particle picture.  Let us now give a useful simplification for free bosons or fermions.
\begin{theorem}\label{th:23}
 Take a state $\rho(t)=U(t)\rho(0)U^{\dagger}(t)$ of $N$ non-interacting bosons or fermions and a measurement operator counting the number of particles in some orthogonal modes $M=\sum_ib^{\dagger}_ib_i$, where $b_i=a(\ket{\phi_i})$. Then, there exist orthonormal single-particle states $\ket{\psi_{\alpha}(t)}$, evolving via the corresponding single-particle Hamiltonian, such that
 \begin{equation}\label{eq:18}
 \tr{\rho(t) M}= \sum_{\alpha}n_{\alpha}\tr{\psi_{\alpha}(t) P},
\end{equation}
where $n_{\alpha}$ are occupation numbers adding up to $N$, $P=\sum_i\ket{\phi_i}\bra{\phi_i}$ and $\psi_{\alpha}(t)=\ket{\psi_{\alpha}(t)}\bra{\psi_{\alpha}(t)}$.
\end{theorem}
}

\bl{
This is proved for \textit{any} $N$ particle state in appendix \ref{app:Proof of Theorem 1}.  Here we will just prove it for the simpler case of an initial state with $N$ bosons or fermions $a^{\dagger n_1}_1...a^{\dagger n_k}_{n_k}\ket{0}$, where $a^{\dagger}_{\alpha}=a^{\dagger}(\ket{\psi_{\alpha}})$, and $\ket{\psi_{\alpha}}$ are some orthonormal single-particle states.  In this case, the states $\ket{\psi_{\alpha}}$ mentioned in the theorem are already given.

\begin{proof}
Expand
\begin{equation}
 a^{\dagger}(\ket{\phi_i})=b^{\dagger}_i=\sum_{\alpha}c_{i,\alpha}a^{\dagger}_{\alpha},
 \end{equation}
where $c_{i,\alpha}$ are complex numbers.  Then
\begin{equation}
\begin{split}
 \tr{\rho(0)\, b^{\dagger}_ib_i} & =\sum_{\alpha}|c_{i,\alpha}|^2\tr{\rho(0)\, a^{\dagger}_{\alpha}a_{\alpha}}\\
 & =\sum_{\alpha}n_{\alpha} |c_{i,\alpha}|^2,
 \end{split}
\end{equation}
where $n_{\alpha}$ is the number of particles in mode $\alpha$.  Next, we use $c_{i,\alpha}=\{a_{\alpha},b^{\dagger}_i\}=\langle \psi_{\alpha}|\phi_i\rangle$ for fermions, or $c_{i,\alpha}=[a_{\alpha},b^{\dagger}_i]=\langle \psi_{\alpha}|\phi_i\rangle$ for bosons, to get
\begin{equation}
 \tr{\rho(0)\, b^{\dagger}_ib_i} =\sum_{\alpha}n_{\alpha}\langle \psi_{\alpha}|\phi_i\rangle\langle \phi_i|\psi_{\alpha}\rangle.
\end{equation}
Therefore,
\begin{equation}
\begin{split}
 \tr{\rho(0)\, M} & =\sum_{\alpha}n_{\alpha}\langle \psi_{\alpha}|P|\psi_{\alpha}\rangle\\
 & = \sum_{\alpha}n_{\alpha}\tr{\psi_{\alpha}P},
 \end{split}
\end{equation}
where $P=\sum_i\ket{\phi_i}\bra{\phi_i}$.
To incorporate the dependence on time, we use $a_{\alpha}(t)=U(t)a_{\alpha}U^{\dagger}(t)$ and
\begin{equation}
 \tr{\rho(0)\, a^{\dagger}_{\alpha}a_{\beta}} = \tr{\rho(t)\, a^{\dagger}_{\alpha}(t)a_{\beta}(t)}.
\end{equation}
The end result is
\begin{equation}
 \tr{\rho(t) M}= \sum_{\alpha}n_{\alpha}\tr{\psi_{\alpha}(t) P}.
\end{equation}
\end{proof}
Notice that linearity of the time average, together with equation (\ref{eq:18}) implies
\begin{equation}
 \tr{\avg{\rho} M}= \sum_{\alpha}n_{\alpha}\tr{\avg{\psi_{\alpha}} P}.
\end{equation}
\subsection{Coarse-grained measurements}}
We can apply this to $\Delta_M(t)$, noting that for the applications we are interested in $\|M\|=N$ when restricted to the $N$ particle subspace.  This occurs, for example, when we are measuring the particle number in a region of space.  Put another way, we take the experimental accuracy of our measurements to be at best $\epsilon N$, where $\epsilon$ is some very small constant.  For equilibration to occur, we need $\Delta_M(t)$ to be small compared to $\epsilon$ most of the time.  We get
\begin{equation}\label{eq:4}
\begin{split}
 \frac{\abs{\tr{\rho(t) M}-\tr{\avg{\rho} M}}}{\|M\|} & = \abs{\tr{\sigma(t) P}-\tr{\avg{\sigma} P}}\\
 & = D_P(\sigma(t),\avg{\sigma}),
 \end{split}
\end{equation}
where $\sigma(t)=\f{1}{N}\sum_j\ket{\psi_j(t)}\bra{\psi_j(t)}$ is a single-particle state.  In words, the $N$ particle problem has been replaced by a single-particle problem in terms of the distinguishability given a single measurement with projectors $P$ and $\openone - P$.

Now recall that equilibration of expectation values does not necessarily imply that equilibration will be observed.  For the examples we look at, the fluctuations in the observed value of $M$, given by $(\tr{\rho(t)M^2}-\tr{\rho(t)M}^2)^{1/2}$, are bounded above by $\sqrt{N}$, which is proved in appendix \ref{app:fluctuations}.  In fact, a large class of fermion systems have time-averaged fluctuations bounded above by $\sqrt{N}$, as seen in appendix \ref{app:fluctuations}.  For large numbers of particles, \bl{(}comparable to $10^{23}$, for example\bl{)} $\sqrt{N}$ is small compared to our experimental precision $\epsilon N$\bl{,} and the fluctuations are not practically observable.  Even for dilute gases with $O(10^4)$ particles, $\sqrt{N}\sim 100$, so the fluctuations are of the order of $1\%$ of the total particle number, which is still quite small.
\newline
\bl{\subsection{Few-mode measurements}
\label{sec:Few-mode measurements}
We are not just restricted to coarse-grained measurements.  We can also discuss measurements involving a few modes.  These could be single-site densities or correlation functions in the setting of lattice models.  Or they could be phase correlators $\mathrm{tr}[\rho(t)a^{\dagger}_ia_j]$, which are typically inferred from time-of-flight measurements \cite{Friesdorf15}.

We will return to this in section \ref{sec:Local equilibration}, where we will see that for a large class of lattice systems any measurement on a small number of modes (small compared to the lattice size) will equilibrate.  And the timescale will be relatively fast.

\section{Example I:\ Particles in a Box.}}
Suppose we have a one dimensional box with a partition at the halfway point (this can be extended to a three dimensional example as shown in appendix \ref{app:Calculations for Fermions in a Box}).  On the left of the partition we have $N$ \bl{fermions or bosons} at zero temperature.  We open the partition at $t=0$, and the observable we focus on is $M$, which counts the \bl{particles} in the left half of the box.
\begin{figure}[ht!]
\centering
    \resizebox{8.0cm}{!}{\includegraphics{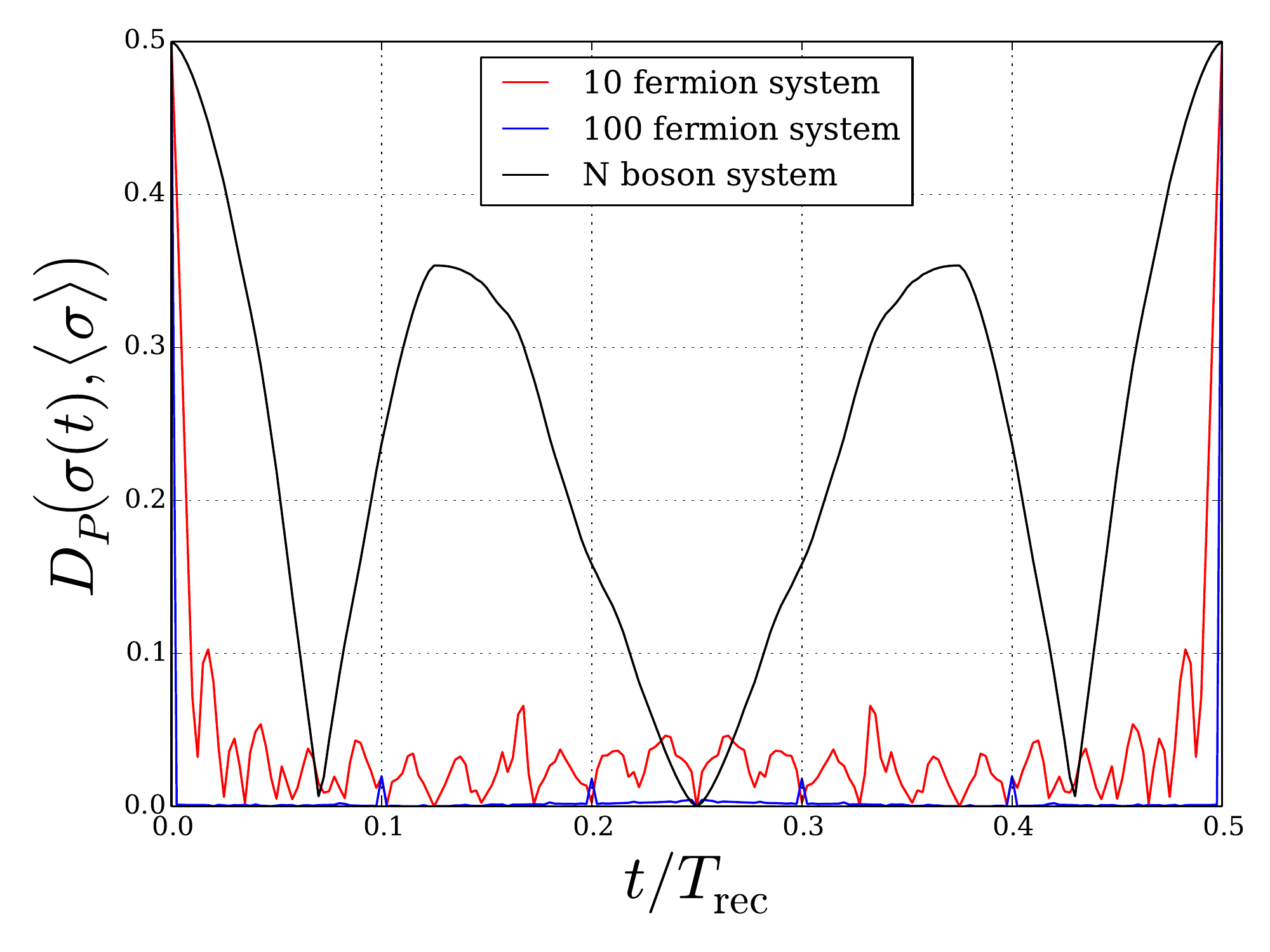}}
    \footnotesize{\caption[Fermions in a box]{\bl{Equilibration of a gas of particles in a box.  The initial state corresponds to $N$ fermions or bosons trapped on one side of a partition, which is removed at $t=0$.  The measurement we consider counts the number of particles on the left side of the box.  Above is a plot of the resulting single-particle distinguishability for the one dimensional example with $N=10$ or $100$ fermions and any number of bosons.  Time is measured in units of the recurrence time, though for the initial state here there is another recurrence at half the recurrence time.  In general, for fermions the equilibration time is $O(1/N)$.  For bosons, the system does not equilibrate, as can be seen from the figure.  These plots were generated using equation (\ref{eq:29}).}\label{fig:fermions}}}
\end{figure}

Using equation (\ref{eq:4}), we can replace this $N$ particle problem by a single-particle one, so
\begin{equation}\label{eq:8}
 \abs{\frac{\tr{\rho(t) M}-\tr{\avg{\rho} M}}{N}}  = D_{P}(\sigma(t),\avg{\sigma}),
\end{equation}
which is plotted in figure \ref{fig:fermions}.  Here $\sigma(t)$ is a state of a free particle in a box, $\avg{\sigma}$ is its time average and $P$ is the projector onto the left hand side of the box.

\bl{
\subsection{Fermions}
First, let us look at the case where the particles are fermions.  The initial state of the $N$ fermion system has all fermions in the left half of the box at temperature zero.  This means that the initial single-particle state $\sigma(0)$ is an equal mixture of the lowest $N$ energy levels of a particle trapped in the left half of a box.  This can be seen from equation (\ref{eq:18}).}

\bl{
The energy eigenstates for a particle in a box are given by
\begin{equation}
 \ket{n}=\int_0^L\!\mathrm{d}x\sqrt{\frac{2}{L}}\sin\left(\frac{\pi nx}{L}\right)\ket{x},
\end{equation}
where $n>0$ is an integer and $L$ is the length of the box.  Similarly, the energy eigenstates for a particle trapped in the left half of the box are given by
\begin{equation}\label{eq:42}
 \ket{\psi_k}=\int_0^{L/2}\!\mathrm{d}x\sqrt{\frac{4}{L}}\sin\left(\frac{2\pi kx}{L}\right)\ket{x},
\end{equation}
where again $k>0$ is an integer.

The initial state of the single-particle system is
\begin{equation}\label{eq:40}
 \sigma(0)=\frac{1}{N}\sum_{k=1}^{N}\ket{\psi_k}\bra{\psi_k},
\end{equation}
with matrix elements $\sigma_{nm}=\bra{n}\sigma(0) \ket{m}$.  Similarly, the matrix elements of the projector onto the left half of the box are $P_{mn}=\bra{m}P \ket{n}$.

Let us look at the distinguishability to see if the system equilibrates.  In figure \ref{fig:fermions}, the distinguishability as a function of time is plotted.  From the plots we can see that, as the number of fermions increases, the average distinguishability gets smaller.  Notice that particles in a box have an exact recurrence time of $T_{\mathrm{rec}}=\f{2\pi}{\nu}$ since the energy levels are $E_n=\nu n^2$, where $n$ is an integer greater than zero, and $\nu=\f{\pi^2}{2mL^2}$.  This is because all phases of density matrix elements in the energy basis $e^{-i(E_n-E_m)t}$ are $1$ at $t=\f{2\pi}{\nu}$.  As in \cite{MLS15}, this means that we need only study the system over the interval $[0,T_{\mathrm{rec}}]$.  In fact, with the particular initial state below, a recurrence actually occurs at $T_{\mathrm{rec}}/2$.

Evaluating the distinguishability at time $t$, we get
\begin{equation}\label{eq:41}
\begin{split}
 D_{P}(\sigma(t),\langle\sigma\rangle) = \abs{\sum_{n\neq m} e^{-i(n^2-m^2)\nu t}\sigma_{nm}P_{mn}}&\\
 = \abs{\sum_{n\neq m} \cos\left[(n^2-m^2)\nu t\right] \sigma_{nm}P_{mn}}&,
 \end{split}
\end{equation}
where we used the fact that $\sigma_{nm}$ and $P_{mn}$ are symmetric under swapping $n$ and $m$ because all vectors and operators here are real.

In appendix \ref{app:Calculations for Fermions in a Box}, we see that the distinguishability can be written as
\begin{equation}\label{eq:29}
D_{P}(\sigma(t),\langle\sigma\rangle)=
 \frac{2}{N}\abs{\sum_{n\ \mathrm{odd}}\sum_{k=1}^N\cos\left[(n^2-4k^2)\nu t\right]f(n,2k)^2},
\end{equation}
where
\begin{equation}
 f(n,2k)^2=\frac{4}{\pi^2}\frac{4k^2}{(n^2-4k^2)^2},
\end{equation}
for $n\neq 2k$.}

\bl{For the system to equilibrate, we need it to spend most of its time indistinguishable from its time-average state.  We see in appendix \ref{app:Calculations for Fermions in a Box}, that the time-average distinguishability satisfies $\langle D_{P}(\sigma(t),\avg{\sigma})\rangle =O(\ln(N)^2/N)$.  Therefore, most of the time the system state is indistinguishable from its time average, provided $N$ is large.

We can also say something about the timescale.  We see in appendix \ref{app:Calculations for Fermions in a Box}, that the timescale for equilibration is at most
\begin{equation}
T_{\mathrm{eq}}= \frac{1}{Na\nu}=\frac{2mL^2}{Na\pi^2}.
\end{equation}
Here $a$ is a small constant that we choose such that the distinguishability at $t=T_{\mathrm{eq}}$ is small:
\begin{equation}
 D_{P}(\sigma(T_{\mathrm{eq}}),\langle\sigma\rangle)  \leq \frac{\pi a}{3}+O\left(\frac{\log(N)^2}{N}\right),
\end{equation}
which is also derived in appendix \ref{app:Calculations for Fermions in a Box}.  Interestingly, the timescale decreases with increasing particle number.

\subsection{Bosons}
The situation for $N$ bosons is simpler.  As they are initially at zero temperature, all $N$ bosons are in the ground state.  The corresponding initial single-particle state $\sigma(0)$ is just the lowest energy state for a particle trapped on the left of the partition.  This does not depend on $N$.  By looking at the plot of the distinguishability in figure \ref{fig:fermions}, it is clear that this system does not equilibrate because the distinguishability is large for most times.

So the behaviour of $N$ bosons is very different from the fermion case.  This is because of the exclusion principle:\ in the fermion case, the fermions had to occupy different energy levels and so the corresponding single-particle state was spread out over many energy levels.  This is not the case for bosons at zero temperature.
}

\section{Example II:\ Bosons after a quench}
\bl{For our} second example, \bl{suppose we have} $N$ bosons \bl{at zero temperature in a one dimensional harmonic trap with frequency $\omega_0$.}
\bl{We will consider what happens after two different quenches.

\subsection{Quench to a square well potential}}
Suppose the Hamiltonian changes suddenly so that the bosons are then confined in a deep square well potential, which we can idealize as a box corresponding to the interval $[-\f{L}{2},\f{L}{2}]$.  Let the measurement operator $M$ count the number of bosons in the central region of the box $[-\f{L}{4},\f{L}{4}]$.  Applying equation (\ref{eq:4}), we see that
\begin{equation}
\label{eq:3}
 \abs{\frac{\tr{\rho(t) M}-\tr{\avg{\rho} M}}{N}} = D_{P}(\psi(t),\avg{\psi}),
\end{equation}
where $\psi(t)=\ket{\psi(t)}\bra{\psi(t)}$ is a pure state of a single-particle and $P$ is the projector onto the central region of the box.

The equilibration timescale has already been estimated for this single-particle system in \cite{MLS15}.  First, the infinite-time average of $D_{P}(\psi(t),\avg{\psi})$ is numerically shown to scale like
\begin{equation}
 \langle D_{P}(\psi(t),\avg{\psi})\rangle \sim \bl{\left(\frac{1}{L}\sqrt{\frac{8\pi}{m\omega_0}}\right)^{\,0.79}},
\end{equation}
where it was assumed that the width of the \bl{initial} wavefunction is small compared to the length of the box, meaning $\sigma=1/\sqrt{2m\omega_0} \ll L$.  So for sufficiently narrow potentials (or sufficiently large boxes), equilibration occurs.  Furthermore, the timescale for equilibration is shown to be \cite{MLS15}
\begin{equation}
T_{\mathrm{eq}}=\frac{L}{\pi}\sqrt{\frac{m}{2\omega_0}}.
\end{equation}

It would be interesting to observe this experimentally.  \bl{In fact, it may be feasible to create square-well potentials in practice:\ in \cite{GSGSH13} a three dimensional cylindrical potential was created to trap a Bose-Einstein condensate, so creating potentials with sharply defined walls may be possible.

\subsection{Quench to a weaker harmonic trap}
Recent experiments have followed the dynamics of Bose gases after a different quench to that of the previous section.  By quickly changing the strength of a harmonic trap, oscillatory behaviour was observed \cite{FCJB14}.  \re{Such behaviour occurred in both the strongly and weakly interacting regimes.  For our purposes, the latter of these regimes is relevant and corresponds to an ideal Bose gas in one dimension.}  In \re{\cite{FCJB14}} the ratio of initial trap frequency $\omega_0$ and post-quench frequency $\omega$ was close to one:\ $\omega_0/\omega\simeq 1.3$.  Here we see equilibration when $\omega_0/\omega$ is much larger than one.

For our observable, let us take the number of bosons in the spatial region $[-l,l]$.  Again, using equation (\ref{eq:4}), we can replace this $N$ particle problem by a single-particle one, so
\begin{equation}
 \abs{\frac{\tr{\rho(t) M}-\tr{\avg{\rho} M}}{N}}  = D_{P}(\psi(t),\avg{\psi}),
\end{equation}
where the $P$ is the projector onto the region $[-l,l]$.  The distinguishability is
\begin{equation}
D_{\psi}(\psi(t),\avg{\psi})= |\tr{P\psi(t)}-\tr{P\av{\psi}}|.
\end{equation}
So we need only see if $\tr{P\psi(t)}$ spends most of its time close to its time average.

The problem is simplified by using the propagator for a harmonic oscillator with frequency $\omega$, given by \cite{Pauli00}
\begin{equation}
\begin{split}
 & K(x,y,t)=\\
 & \sqrt{\frac{m\omega}{2\pi i \sin(\omega t)}}\exp\left[-\frac{m\omega((x^2+y^2)\cos(\omega t)-2xy)}{2i \sin(\omega t)}\right],
 \end{split}
\end{equation}
which leads to the expression
\begin{equation}
 \begin{split}
  & \tr{P\psi(t)}=\\
  & \int\!\mathrm{d}y_1\mathrm{d}y_2\int_{-l}^{l}\!\mathrm{d}x\,\psi^*(y_1)K^*(x,y_1,t)K(x,y_2,t)\psi(y_2).
 \end{split}
\end{equation}
As $\psi(x)$ is a Gaussian wavefunction, the $y_1$ and $y_2$ integrals are straightforward, leading to
\begin{equation}
 \begin{split}
  \tr{P\psi(t)}= \frac{1}{\sqrt{\pi}} \int_{-l\sqrt{\alpha(t)}}^{l\sqrt{\alpha(t)}}\!\mathrm{d}x\,e^{-x^2},
 \end{split}
\end{equation}
where
\begin{equation}\label{eq:190}
 \alpha(t)=\frac{m\omega_0}{\gamma^2\sin^2(\omega t)+\cos^2(\omega t)},
\end{equation}
with $\gamma=\omega_0/\omega$.  Next we use the approximation for the error function \cite{Winitzki08}
\begin{equation}
\begin{split}
 \mathrm{erf}(x) & =\frac{2}{\sqrt{\pi}}\int^x_0\!\mathrm{d}t\,e^{-t^2}\\
 & \simeq \mathrm{sgn}(x)\sqrt{1-\exp\left[-x^2\frac{\f{4}{\pi}+bx^2}{1+bx^2}\right]},
 \end{split}
\end{equation}
where the maximum error for any $x$ is around $0.00012$, and $b\simeq 0.147$.  The result is that
\begin{equation}\label{eq:50}
 \begin{split}
  \tr{P\psi(t)}\simeq \sqrt{1-\exp\left[-\alpha(t)l^2\frac{\f{4}{\pi}+\alpha(t)b\,l^2}{1+\alpha(t)b\,l^2}\right]}.
 \end{split}
\end{equation}
Notice that there are four independent parameters that matter:\ $l$, which controls the width of the interval the measurement looks at; $\omega$, which is the frequency of the trap after the quench; $\gamma=\omega_0/\omega$, which is the ratio of trap strengths before and after the quench; and $m\omega_0$, which determines the width of the initial state.  A natural starting point is to choose $l$ so that the initial state is almost entirely contained in $[-l,l]$, so we can fix $l^2=4/(m\omega_0)$.

As we can see from figure \ref{fig:bosons}, as $\gamma$ becomes bigger and bigger, $\tr{P\psi(t)}$ spends most of its time close to zero.  So for very large $\gamma$, equilibration occurs.  In fact, we can see directly from equations (\ref{eq:190}) and (\ref{eq:50}) that, as $\gamma$ tends to infinity, $\tr{P\psi(t)}$ tends to zero.  This holds for all times, except when $\omega t=n\pi$, with $n\in\mathbb{Z}$.


In \cite{FCJB14}, oscillatory behaviour was seen at $\gamma=1.3$.  Here, this value of $\gamma$ does not lead to any significant departure from the initial state, as seen in figure \ref{fig:bosons}.  The reason for this difference is that in \cite{FCJB14} the initial states were at non-zero temperature.  Here, we are initially at zero temperature, and we see oscillations at higher values of $\gamma$.
\begin{figure}[ht!]
\centering
    \resizebox{8.0cm}{!}{\includegraphics{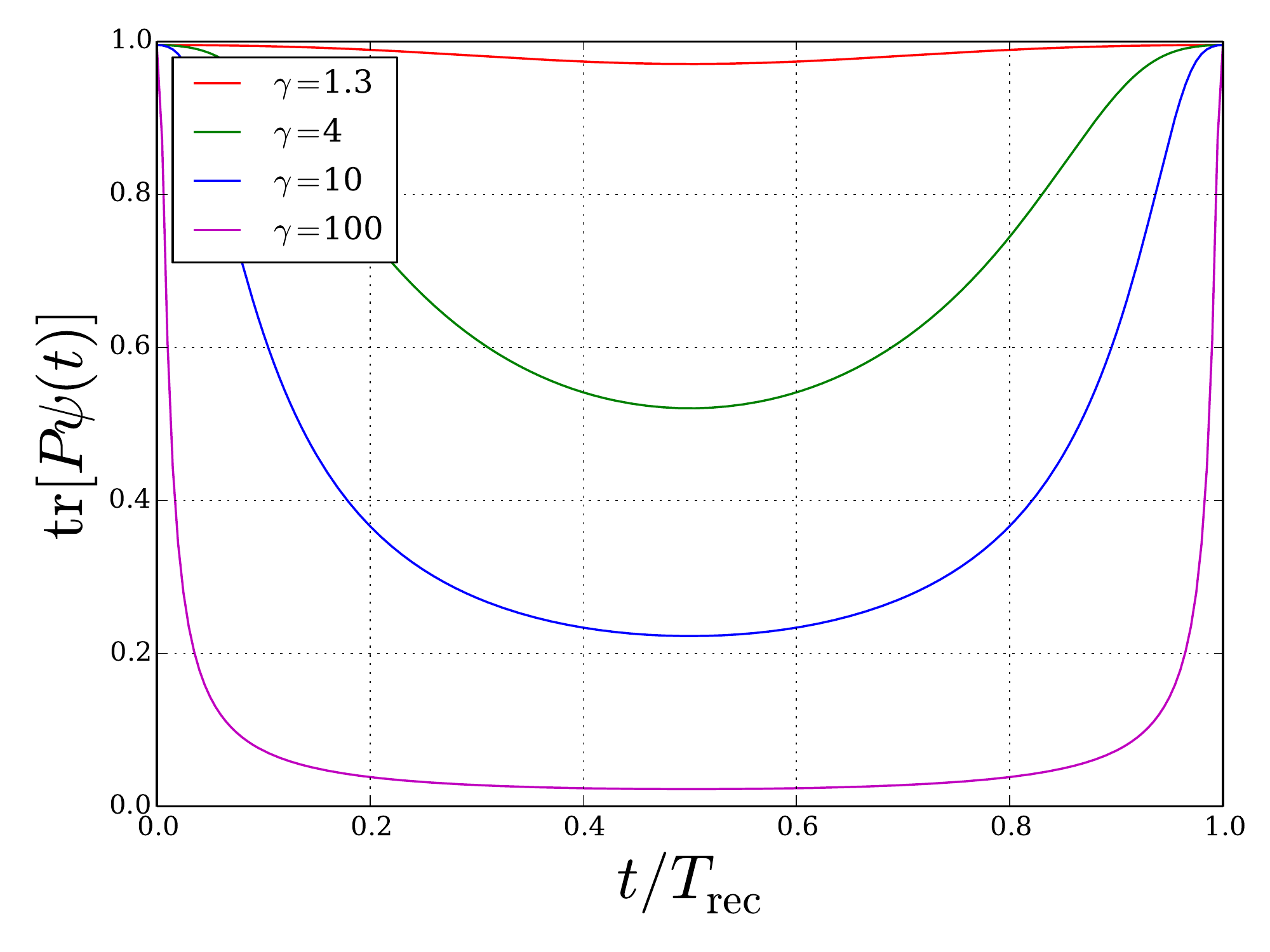}}
    \footnotesize{\caption[Bosons in a trap.]{\bl{Equilibration of bosons in a harmonic trap.  Initially we have $N$ bosons in the ground state of a harmonic trap with frequency $\omega_0$.  The trap strength is then quenched to $\omega$.  The measurement we consider counts the number of bosons in a window of width $4/\sqrt{m\omega_0}$.  Here we have plots of the corresponding single-particle quantity $\tr{P\psi(t)}$ for different values of $\gamma=\omega_0/\omega$.  The value $\gamma=1.3$ corresponds to that from \cite{FCJB14}.  We see oscillatory behaviour for $\gamma$ approximately between $4$ and $10$.  We see that, as $\gamma$ becomes larger, $\tr{P\psi(t)}$ is small most of the time, and the system equilibrates.  These plots were generated using equation (\ref{eq:50}).}\label{fig:bosons}}}
\end{figure}

To estimate the equilibration time when equilibration does occur, we estimate how long it takes for $\tr{P\psi(t)}$ to reach $p\ll 1$.  Using equation (\ref{eq:50}) and $\log(1-p^2)\simeq -p^2$, we get
\begin{equation}
 \alpha(t)l^2\frac{\f{4}{\pi}+\alpha(t)b\,l^2}{1+\alpha(t)b\,l^2}\simeq p^2.
\end{equation}
Since $p$ is small, this requires $\alpha(t)l^2$ to be small.  Using the earlier choice $l^2=4/(m\omega_0)$, we get
\begin{equation}
 \alpha(t)l^2\frac{4}{\pi}=\frac{16/\pi}{\gamma^2\sin^2(\omega t)+\cos^2(\omega t)}\simeq p^2.
\end{equation}
For large $\gamma$, this is satisfied at
\begin{equation}
 T_{\mathrm{eq}}\simeq \frac{4}{\sqrt{\pi}\gamma\omega p}=\frac{4}{\sqrt{\pi}\omega_0 p},
\end{equation}
where we assumed that $t$ was small compared to $1/\omega$, and used $\sin(x)\simeq x$ for small $x$.
}

\section{Equilibration in General}
\bl{The examples we looked at were encouraging, but a pressing question is whether one can say anything more general.  The answer is actually yes:\ here we will see general estimates for the equilibration timescale of gases with negligible interactions.}  The starting point is \bl{again to} replace the $N$ particle problem by a single-particle problem.  Then we can \bl{use a single-particle equilibration result, which builds on} previous results \cite{SF12,RK12,MGLFS14}.

\bl{First, let us state the single-particle equilibration result.  We will need to take account of degenerate energy gaps.  These occur when two different energy gaps are equal:\ $E_i-E_j=E_k-E_l$, when $E_i\neq E_k$ and $E_i\neq E_j$.  The degeneracy of the most degenerate gap is denoted by $D_G$.  If all gaps are different, $D_G=1$.}  For a particle in a box there are some degenerate energy gaps, though the addition of an inhomogeneous potential $V(x)$ would \bl{generally} change this.  \bl{The harmonic oscillator also has many degenerate energy gaps.}  Typically, \bl{however, these are very special cases, and} we expect most Hamiltonians would have few degenerate energy gaps.

\bl{
\begin{theorem}\label{th:54}
Suppose we have a single-particle system with a $d$ dimensional state space.  Let $A$ be an operator with operator norm $\|A\|$, and let $\sigma(t)$ be a state unitarily evolving via a Hamiltonian $H$.  Denote the infinite-time average of $\sigma(t)$ by $\av{\sigma}$.  Assuming that we can make the density of states approximation, meaning we replace $\sum_E$ by $\int \!\mathrm{d}E\, n(E)$, where $n(E)$ is the density of states, we get
\begin{equation}\label{eq:5}
\begin{split}
\frac{\av{ |\tr{\sigma(t)A} - \tr{ \av{\sigma} A }|^2 }_T}{\norm{A}^2} \leq\frac{c_1}{\de}\left[D_G + c_2\frac{n_{\max}d}{T}\right],
\end{split}
\end{equation}
where $\av{\cdot}_T$ denotes the time average over $[0,T]$, and we have constants $c_1=e\sqrt{\pi}/2$ and $c_2=4\sqrt{\pi}$.  Also, $n_{\max}=\max_En(E)$, and the effective dimension of the state $\sigma(t)$ is defined by
\begin{equation}
 \frac{1}{d_{\mathrm{eff}}}=\sum_{E}\left(\tr{\sigma(0) P_{E}}\right)^2,
\end{equation}
where $P_E$ is the projector onto the energy eigenspace corresponding to energy $E$.
\end{theorem}
}

Equilibration of the expectation value of $A$ occurs \bl{provided} the right hand side of equation (\ref{eq:5}) is sufficiently small.  \bl{As $T\rightarrow \infty$, equilibration is guaranteed if $c_1D_G/\de\ll 1$.  The effective dimension $\de$ measures how spread out over energy levels the initial state is.  If $\de$ is very large we expect equilibration to occur.

But we can also estimate the timescale:\ the equilibration timescale can be bounded above by the smallest $T$ such that the right hand side of equation (\ref{eq:5}) is small.  In other words, when equilibration occurs, we get an upper bound for the timescale:
\begin{equation}
 T_{\mathrm{eq}}\propto \frac{n_{\max}d}{\de}.
\end{equation}
The main task now is to use this single-particle equilibration result to find timescales for $N$ particle systems.
}

\bl{
\subsection{Coarse-grained measurements}
Let us start with coarse-grained measurements.  We will see that equilibration of coarse-grained observables generally occurs much quicker than what we would expect based on previous timescale bounds from \cite{SF12,MGLFS14}.

By} mapping an $N$ particle problem to the single-particle picture via equation (\ref{eq:4}), we want to bound
\begin{equation}
\begin{split}
 \avg{\abs{\frac{\tr{\rho(t) M}-\tr{\avg{\rho} M}}{N}}}_{T}  = \avg{D_{P}(\sigma(t),\avg{\sigma})}_T&\\
 =\av{ |\tr{\sigma(t)P} - \tr{ \avg{\sigma} P }| }_T&\\
 \leq \sqrt{\av{ \left(\tr{\sigma(t)P} - \tr{ \avg{\sigma} P }\right)^2 }_T}&\\
  \leq\sqrt{\frac{c_1}{\de}\left[D_G + c_2\frac{n_{\max}d}{T}\right]}&,
 \end{split}
\end{equation}
where the \bl{third line} follows from concavity of the square root.  \bl{The last line follows from the result for single-particle equilibration from the last section, namely equation (\ref{eq:5}).}

To \bl{see} whether equilibration occurs at all, let $T\to \infty$ to get the infinite-time average.  \bl{And so we must estimate $\f{1}{\de}$.  Defining $N_E$ to be the operator that counts the number of particles in energy level $E$, we get
\begin{equation}
\begin{split}
 \frac{1}{d_{\mathrm{eff}}} & =\sum_{E}(\tr{\sigma P_E})^2 \\
 & = \frac{1}{N^2}\sum_{E}(\tr{\rho N_E})^2\\
 & = \sum_{E}\left(\frac{n_E}{N}\right)^2,
 \end{split}
\end{equation}
where $n_E=\tr{\rho N_E}$.  Getting the second line used equation (\ref{eq:18}) from theorem \ref{th:23}.  So we see that $\f{1}{\de}$ is extremely small if the state is spread over many energy levels.}

\bl{Consider} $N$ fermions or \bl{the special case of} $N$ bosons in \bl{orthogonal} modes\bl{.  Then} the resulting single-particle density operator $\sigma(t)$ is an equal mixture of $N$ orthogonal states.  \bl{In that case, $1/d_{\mathrm{eff}}\leq\f{n_d}{N^2}\sum_En_E\leq n_d/N$,  where $n_d$ is the degeneracy of the most degenerate energy level.}  As a result,
\begin{equation}
 \bl{\avg{\abs{\frac{\tr{\rho(t) M}-\tr{\avg{\rho} M}}{N}}}
\leq  \sqrt{\frac{c_1 n_d D_G}{N}}}.
\end{equation}

So the bottom line is that for the coarse measurements considered here, equilibration occurs very generally \bl{despite the fact that these are} non-interacting bosons or fermions. 

We can also say something substantial \bl{about the equilibration timescale}.

\bl{We can always restrict our attention to} $d$ energy levels of the \bl{corresponding} single-particle system, which may require an energy cutoff.  And suppose $d$ is bounded above by a polynomial in $N$.  This depends on the state $\sigma(t)$ and so ultimately on the state of each of the $N$ bosons or fermions.  For example, for the calculations with fermions equilibrating in appendix \ref{app:Calculations for Fermions in a Box}, we effectively took a cutoff with $d=O(N)$.  \bl{In fact, for the bosonic examples, $d$ was independent of $N$.}  For lattice systems this is particularly natural if there is a constant density of particles, then $\bl{d\propto V\propto N}$, where \bl{$V$} is the number of lattice sites.

Next, \bl{we estimate $n_{\max}$, which is often polynomial in $d$, and hence $N$.  F}or example, \bl{$n_{\max}\sim d^3$} for a system whose energy levels go like $E_n\propto 1/n^2$, similar to the energies for bound states in a Coulomb potential.  Notice that this is a system we would expect to have very many small gaps.  Conversely, when the energy level spacings grow with $d$ we would expect better behaviour.  For example, when $E_n\propto n^2$, one gets \bl{$n_{\max}\sim 1$}.

\bl{Putting this all together, if equilibration occurs, the timescale is typically}
\begin{equation}
T_{\mathrm{eq}}\in O(N^k)
\end{equation}
for some positive integer $k$.  This is far better than the bounds of \cite{SF12,MGLFS14,RK12}, which were exponential in $N$ for physical systems.  Of course, how \bl{$n_{\max}$} scales with $d$ and how $d$ scales with $N$ depend on the system in question, but neither of the requirements above appear unnaturally restrictive.

\re{It is also interesting that each of \cite{Cramer12,BCHHKM12,MRA13} found equilibration timescales that were polynomial (or faster) in the number of particles.  In contrast to the setting considered here, these results involved averaging over Hamiltonians with respect to the global unitary Haar measure.  Because of this, it is not clear how to interpret the implications of \cite{Cramer12,BCHHKM12,MRA13} for equilibration of \textit{local} Hamiltonian systems.  Nevertheless, \cite{Cramer12,BCHHKM12,MRA13} do say something about equilibration timescales of fully interacting models, which is very interesting.}

\bl{\subsection{Local equilibration}\label{sec:Local equilibration}}
\bl{We can also look at equilibration of non-interacting lattice models.  This would include the free superfluid regime of the Bose-Hubbard model, for example.  We consider local few-mode measurements, where few means that the number of modes is small compared to the number of lattice sites.  This setting includes all measurements in some small region of the lattice or correlation functions over long distances.  It also includes phase correlators, which are important in ultracold atom systems.

We will state the single-mode result first.  This relies on the Hamiltonian being some form of local (not necessarily nearest-neighbour) hopping Hamiltonian:\ the tight-binding model is one example.

To make the formulas easier to read, we will assume that the maximum energy level degeneracy $n_d$ and the number of modes per site are both one.  In the proofs of these results in appendix \ref{sec:Free lattice models} we allow other values of these quantities.
\begin{corollary}\label{th:76}
Take a free lattice model, and assume we can make the density of states approximation, as in theorem \ref{th:54}. Let $M=a^{\dagger}(\ket{\phi})a(\ket{\phi})$, where $\ket{\phi}$ is a single-particle state localized on at most $l$ sites (which need not be near each other).  Then we have
\begin{equation}\label{eq:678}
\begin{split}
\av{\abs{\tr{\rho(t) M}-\tr{\avg{\rho} M}}}_T\leq l\sqrt{c_1\left[\frac{D_G}{d} + c_2\frac{n_{\max}}{T}\right]},
 \end{split}
\end{equation}
where $d$ is the dimension of the corresponding single-particle Hilbert space, and we have constants $c_2=4\sqrt{\pi}$ and $c_1=(e\sqrt{\pi}/2)$.  Also, $n_{\max}=\max_En(E)$, where $n(E)$ is the density of states.

For bosons, we needed to assume that the initial state has at most one boson in each mode.  Otherwise, the same result holds, but with an extra factor on the right hand side given by the maximum number of bosons in a given mode.
\end{corollary}
This is proved in appendix \ref{sec:Free lattice models}.  Again, we see equilibration provided the right hand side of equation (\ref{eq:678}) is small.  We will estimate the equilibration timescale below corollary \ref{cor:143}.  First, let us discuss some consequences of this result.

A simple consequence is that phase correlators equilibrate.  Phase correlators are expectation values like $\mathrm{tr}[\rho(t)a^{\dagger}_{\vec{x}}a_{\vec{y}}]$, where $\vec{x}$ and $\vec{y}$ denote lattice sites.  (There may be several modes at each lattice site, but for simplicity of notation, we have assumed that there is just one.)  Using
\begin{equation}
 a^{\dagger}_{\vec{x}}a_{\vec{y}}=\frac{1}{2}\left(d^{\dagger}_1d_1-d^{\dagger}_2d_2-id^{\dagger}_3d_3+id^{\dagger}_4d_4\right),
\end{equation}
where
\begin{equation}
 \begin{split}
  d_1 & =\f{1}{\sqrt{2}}(a_{\vec{x}}+a_{\vec{y}})\\
  d_2 & =\f{1}{\sqrt{2}}(a_{\vec{x}}-a_{\vec{y}})\\
  d_3 & =\f{1}{\sqrt{2}}(a_{\vec{x}}+ia_{\vec{y}})\\
  d_4 & =\f{1}{\sqrt{2}}(a_{\vec{x}}-ia_{\vec{y}}),
 \end{split}
\end{equation}
we can express $\mathrm{tr}[\rho(t)a^{\dagger}_{\vec{x}}a_{\vec{y}}]$ in terms of single-mode densities.  And so via the triangle inequality, we can upper bound the time average of $|\mathrm{tr}[\rho(t)a^{\dagger}_{\vec{x}}a_{\vec{y}}]-\mathrm{tr}[\rho(t)a^{\dagger}_{\vec{x}}a_{\vec{y}}]|$ using corollary \ref{th:76}.

Interestingly, these results apply to a vast range of initial states $\rho(0)$.  This means that one could perform a huge variety of quenches to a free lattice system, and the equilibration results here and timescale bounds (which we will discuss below) apply.

Before discussing timescales, we can build on corollary \ref{th:76} further, getting the corollary below, which is proved in appendix \ref{sec:From single-mode to multi-mode measurements}.  We only prove the fermionic result, as the bosonic result is essentially the same.

\begin{corollary}\label{cor:143}
Take a free lattice model, and let $M$ be an operator on $l$ sites.  Suppose the initial state $\rho(0)$ is Gaussian and satisfies $[\rho(0),N]=0$, where $N$ is the total number operator.  (This is still quite general, though it rules out BCS states, for example.)  Then we get
\begin{equation}
\av{\abs{\tr{\rho(t)M}-\tr{\av{\rho}M}}}_T\leq 2^{l+2}ml^2\sqrt{c_1\left[\frac{D_G}{d}+ c_2\frac{n_{\max}}{T}\right]},
\end{equation}
where $d$ is the dimension of the corresponding single-particle Hilbert space, and we have constants $c_2=4\sqrt{\pi}$ and $c_1=(e\sqrt{\pi}/2)$.  Also, $n_{\max}=\max_En(E)$, where $n(E)$ is the density of states.  Finally, $m$ is the maximum coefficient of $M$ when $M$ is expanded in an operator basis of Majorana fermion operators.
\end{corollary}
Typically $m$ will be order one, which is the case for correlation functions, for example.  Therefore, as long as the number of lattice sites that the measurement acts on is quite small, equilibration will also occur for free lattice systems.

Furthermore, we can use these results to upper bound the equilibration timescale.  From corollary \ref{th:76} and corollary \ref{cor:143}, the upper bound for the equilibration timescale scales like $T_{\mathrm{eq}} \propto n_{\max}$.  So it remains to estimate $n_{\max}$.  In appendix \ref{eq:Density of states for lattice models}, we show that for these lattice models, we can effectively take $n_{\max}\propto V$, where $V$ is the number of lattice sites.  Therefore, we get
\begin{equation}
 T_{\mathrm{eq}} \propto V.
\end{equation}
In particular, for one dimensional systems, we get $T_{\mathrm{eq}} \propto L$, where $L$ is the length of the system.

The scaling with system size is quite significant.  Previous bounds \cite{SF12,MGLFS14} were exponential in the system size, whereas here we get something linear.  Furthermore, in the one-dimensional case, the scaling is optimal.  This can be seen from Lieb-Robinson bounds \cite{LR72}, which imply that the time it takes for information to propagate appreciably from one region to another increases linearly with the distance between the regions.  So in one dimension $T_{\mathrm{eq}} \propto L$ is the best we can hope for.  The only possibility for better scaling is if one restricts the set of initial states under consideration.  A good example of such results for special states appeared in \cite{CDEO08}.
}

\section{Discussion and Outlook.}
Finding the timescale involved in equilibration is an important problem in physics, especially in light of recent advances in experiments with mesoscopic quantum systems \cite{EFG14,LGS15}.  The \bl{timescale} results \bl{here} required us to \bl{restrict our attention to a subclass of measurements, which} are physically sensible for macroscopic \bl{or mesoscopic} systems.  \bl{We focused on the regime of negligible interactions, which includes Luttinger liquids and the Hubbard model in the free superfluid regime.  First,} we found example equilibration timescale bounds for gases of bosons and fermions.  \bl{We also} saw that equilibration occurs quite generally in this setting \bl{of very weak interactions} and is very fast compared to the best known general bounds on the equilibration time.

From here the outlook is promising:\ a natural next step is to extend these results to quasi-free systems, where the Hamiltonian is quadratic in terms of creation and annihilation operators but does not conserve particle number.  Such models arise in the theory of superconductivity.  Other options are to extend the results to interacting models via perturbation theory or to look at equilibration in terms of fermionic or bosonic generating functions \cite{Osborne15}.
\vspace{0.1cm}
\acknowledgments
The author is very grateful to Tobias Osborne, Mathis Friesdorf, Jens Eisert, Marcel Goihl, David Reeb and Shane Dooley for helpful discussions and comments.  \bl{The author also thanks the anonymous referees for useful suggestions and comments.}  This work was supported by the ERC grants QFTCMPS and SIQS, and by the cluster of excellence EXC201 Quantum Engineering and Space-Time Research.  The publication of this article was funded by the Open Access Fund of the Leibniz Universit{\"a}t Hannover.
\vfill
\bibliographystyle{unsrt}

\begin{thebibliography}{10}

\bibitem{EFG14}
J.~Eisert, M.~Friesdorf, and C.~Gogolin.
\newblock {Quantum many-body systems out of equilibrium}.
\newblock {\em Nature Physics}, 11:124, 2015.

\bibitem{LGS15}
T.~Langen, R.~Geiger, and J~Schmiedmayer.
\newblock Ultracold atoms out of equilibrium.
\newblock {\em Annual Review of Condensed Matter Physics}, 6(1):201--217, 2015.

\bibitem{Reimann08}
P.~Reimann.
\newblock Foundation of statistical mechanics under experimentally realistic
  conditions.
\newblock {\em Phys. Rev. Lett.}, 101:190403, 2008.

\bibitem{LPSW09}
N.~Linden, S.~Popescu, A.~J. Short, and A.~Winter.
\newblock Quantum mechanical evolution towards thermal equilibrium.
\newblock {\em Phys. Rev. E}, 79:061103, 2009.

\bibitem{Short10}
A.~J. Short.
\newblock Equilibration of quantum systems and subsystems.
\newblock {\em New Journal of Physics}, 13(5):053009, 2011.

\bibitem{SF12}
A.~J. Short and T.~C. Farrelly.
\newblock Quantum equilibration in finite time.
\newblock {\em New Journal of Physics}, 14(1):013063, 2012.

\bibitem{RK12}
P.~Reimann and M.~Kastner.
\newblock Equilibration of isolated macroscopic quantum systems.
\newblock {\em New Journal of Physics}, 14(4):043020, 2012.

\bibitem{MGLFS14}
A.~S.~L. Malabarba, L.~P. Garc\'{i}a-Pintos, N.~Linden, T.~C. Farrelly, and
  A.~J. Short.
\newblock Quantum systems equilibrate rapidly for most observables.
\newblock {\em Phys. Rev. E}, 90:012121, 2014.

\bibitem{GHT13}
S.~Goldstein, T.~Hara, and H.~Tasaki.
\newblock Time scales in the approach to equilibrium of macroscopic quantum
  systems.
\newblock {\em Phys. Rev. Lett.}, 111:140401, 2013.

\bibitem{CDEO08}
M.~Cramer, C.~M. Dawson, J.~Eisert, and T.~J. Osborne.
\newblock Exact relaxation in a class of nonequilibrium quantum lattice
  systems.
\newblock {\em Phys. Rev. Lett.}, 100:030602, 2008.

\bibitem{CE10}
M.~Cramer and J.~Eisert.
\newblock A quantum central limit theorem for non-equilibrium systems:\ exact
  local relaxation of correlated states.
\newblock {\em New Journal of Physics}, 12(5):055020, 2010.

\bibitem{GLMSW15}
L.~P. Garc\'{i}a-Pintos, N.~Linden, A.~S.~L. Malabarba, A.~J. Short, and
  A.~Winter.
\newblock {Equilibration time scales of physically relevant observables}.
\newblock arXiv:1509.05732, 2015.

\bibitem{UWE12}
C.~Ududec, N.~Wiebe, and J.~Emerson.
\newblock Information-theoretic equilibration:\ the appearance of
  irreversibility under complex quantum dynamics.
\newblock {\em Phys. Rev. Lett.}, 111:080403, 2013.

\bibitem{VZ12}
Vinayak and M.~Žnidarič.
\newblock Subsystem dynamics under random {H}amiltonian evolution.
\newblock {\em Journal of Physics A: Mathematical and Theoretical},
  45(12):125204, 2012.

\bibitem{Cramer12}
M.~Cramer.
\newblock Thermalization under randomized local {H}amiltonians.
\newblock {\em New Journal of Physics}, 14(5):053051, 2012.

\bibitem{BCHHKM12}
F.~G. S.~L. Brand\~ao, P.~\ifmmode \acute{C}\else
  \'{C}\fi{}wikli\ifmmode~\acute{n}\else \'{n}\fi{}ski, M.~Horodecki,
  P.~Horodecki, J.~K. Korbicz, and M.~Mozrzymas.
\newblock Convergence to equilibrium under a random {H}amiltonian.
\newblock {\em Phys. Rev. E}, 86:031101, 2012.

\bibitem{MRA13}
L.~Masanes, A.~Roncaglia, and A.~Ac\'{i}n.
\newblock Complexity of energy eigenstates as a mechanism for equilibration.
\newblock {\em Phys. Rev. E}, 87:032137, 2013.

\bibitem{GE15}
C.~Gogolin and J.~Eisert.
\newblock Equilibration, thermalisation, and the emergence of statistical
  mechanics in closed quantum systems.
\newblock {\em Reports on Progress in Physics}, 79(5):056001, 2016.

\bibitem{RSLS15}
B.~Rauer, T.~Schweigler, T.~Langen, and J.~Schmiedmayer.
\newblock {Does an isolated quantum system relax?}
\newblock arXiv:1504.04288, 2015.

\bibitem{KWW06}
T.~Kinoshita and D.~S. Wenger, T.~Weiss.
\newblock A quantum {N}ewton's cradle.
\newblock {\em Nature}, 440:900, 2006.

\bibitem{Bocchieri57}
P.~Bocchieri and A.~Loinger.
\newblock Quantum recurrence theorem.
\newblock {\em Phys. Rev.}, 107:337, 1957.

\bibitem{Schulman78}
L.~S. Schulman.
\newblock Note on the quantum recurrence theorem.
\newblock {\em Phys. Rev. A}, 18:2379, 1978.

\bibitem{NC00}
M.~A. Nielsen and I.~L. Chuang.
\newblock {\em Quantum Computation and Quantum Information}.
\newblock Cambridge University Press, Cambridge, 2000.

\bibitem{BR97}
O.~Bratteli and D.~Robinson.
\newblock {\em Operator Algebras and Quantum Statistical Mechanics, volumes 1
  and 2}.
\newblock Springer, Berlin, 1997.

\bibitem{Friesdorf15}
M.~Friesdorf.
\newblock {\em Closed Quantum Many-body Systems out of Equilibrium}.
\newblock PhD thesis, Freie Universit{\"a}t Berlin, 2015.

\bibitem{MLS15}
A.~S.~L. Malabarba, N.~Linden, and A.~J. Short.
\newblock Rapid spatial equilibration of a particle in a box.
\newblock {\em Phys. Rev. E}, 92:062128, 2015.

\bibitem{GSGSH13}
A.~L. Gaunt, T.~F. Schmidutz, I.~Gotlibovych, R.~P. Smith, and Z.~Hadzibabic.
\newblock Bose-{E}instein condensation of atoms in a uniform potential.
\newblock {\em Phys. Rev. Lett.}, 110:200406, 2013.

\bibitem{FCJB14}
B.~Fang, G.~Carleo, A.~Johnson, and I.~Bouchoule.
\newblock Quench-induced breathing mode of one-dimensional {B}ose gases.
\newblock {\em Phys. Rev. Lett.}, 113:035301, 2014.

\bibitem{Pauli00}
W.~Pauli.
\newblock {\em Wave Mechanics:\ Volume 5 of Pauli Lectures on Physics}.
\newblock Dover, New York, 2000.

\bibitem{Winitzki08}
S.~Winitzki.
\newblock A handy approximation for the error function and its inverse.
\newblock {\em Unpublished}
  http://sites.google.com/site/winitzki/sergei-winitzkis-files/erf-approx.pdf,
  2008.

\bibitem{LR72}
E.~H. Lieb and D.~W. Robinson.
\newblock The finite group velocity of quantum spin systems.
\newblock {\em Communications in Mathematical Physics}, 28(3):251--257, 1972.

\bibitem{Osborne15}
T.~J. Osborne.
\newblock {P}rivate communication, 2015.

\bibitem{Knapp}
M.~P. Knapp.
\newblock Sines and cosines of angles in arithmetic progression.
\newblock {\em Mathematics Magazine}, 82(5):371--372, 2009.

\bibitem{Horn85}
R.~A. Horn and C.~R. Johnson.
\newblock {\em Matrix Analysis}.
\newblock Cambridge University Press, Cambridge, 1985.

\bibitem{FWBSE14}
M.~Friesdorf, A.~H. Werner, W.~Brown, V.~B. Scholz, and J.~Eisert.
\newblock Many-body localization implies that eigenvectors are matrix-product
  states.
\newblock {\em Phys. Rev. Lett.}, 114:170505, 2015.

\bibitem{Greplova13}
E.~Greplov{\'a}.
\newblock Quantum information with {F}ermionic {G}aussian states.
\newblock Master's thesis, Ludwig-Maximilians-Universit{\"a}t M{\"u}nchen,
  2013.

  \end{thebibliography}

\appendix
\section{Proof of Theorem \ref{th:23}}
\label{app:Proof of Theorem 1}
\begin{reptheorem}{th:23}
 Take a state $\rho(t)=U(t)\rho(0)U^{\dagger}(t)$ of $N$ non-interacting bosons or fermions and a measurement operator counting the number of particles in some modes $M=\sum_ib^{\dagger}_ib_i$, where $b_i=a(\ket{\phi_i})$. Then there exist orthonormal single-particle states $\ket{\psi_{\alpha}(t)}$, evolving via the corresponding single-particle Hamiltonian, such that
 \begin{equation}
 \tr{\rho(t) M}= \sum_{\alpha}n_{\alpha}\tr{\psi_{\alpha}(t) P},
\end{equation}
where $n_{\alpha}$ are occupation numbers adding up to $N$, $P=\sum_i\ket{\phi_i}\bra{\phi_i}$ and $\psi_{\alpha}(t)=\ket{\psi_{\alpha}(t)}\bra{\psi_{\alpha}(t)}$.
\end{reptheorem}

\begin{proof}
There exists a complete set of independent annihilation operators $a_{\alpha}=a(\ket{\psi_{\alpha}})$, such that
 \begin{equation}\label{eq:9}
 \tr{\rho(0)\, a^{\dagger}_{\alpha}a_{\beta}} = 0\ \mathrm{if}\ \alpha\neq\beta.
 \end{equation}
To see this, take any complete set of independent annihilation operators $\bl{d}_{\alpha}$.  The matrix $\tr{\rho\, \bl{d}^{\dagger}_{\alpha}\bl{d}_{\beta}}=C_{\alpha\beta}$ is Hermitian.  So there exists a unitary $U_{\alpha\beta}$, such that in terms of
\begin{equation}
 a^{\dagger}_{\alpha}=\sum_{\beta}U_{\alpha\beta}\bl{d}^{\dagger}_{\beta}
\end{equation}
$\tr{\rho(0)\, a^{\dagger}_{\alpha}a_{\beta}}$ is diagonal.  \bl{So it makes sense to work with the $a_{\alpha}$, which determine the states $\ket{\psi_{\alpha}(0)}=\ket{\psi_{\alpha}}$ mentioned in the statement of the theorem via $a_{\alpha}=a(\ket{\psi_{\alpha}})$.}  Next, expand
\begin{equation}
 a^{\dagger}(\ket{\phi_i})=b^{\dagger}_i=\sum_{\alpha}c_{i,\alpha}a^{\dagger}_{\alpha},
 \end{equation}
where $c_{i,\alpha}$ are complex numbers.  Then
\begin{equation}
\begin{split}
 \tr{\rho(0)\, b^{\dagger}_ib_i} & =\sum_{\alpha}|c_{i,\alpha}|^2\tr{\rho(0)\, a^{\dagger}_{\alpha}a_{\alpha}}\\
 & =\sum_{\alpha}n_{\alpha} |c_{i,\alpha}|^2,
 \end{split}
\end{equation}
where $n_{\alpha}$ is the number of particles in mode $\alpha$.  Next, we use $c_{i,\alpha}=\{a_{\alpha},b^{\dagger}_i\}=\langle \psi_{\alpha}|\phi_i\rangle$ for fermions, or $c_{i,\alpha}=[a_{\alpha},b^{\dagger}_i]=\langle \psi_{\alpha}|\phi_i\rangle$ for bosons, to get
\begin{equation}
 \tr{\rho(0)\, b^{\dagger}_ib_i} =\sum_{\alpha}n_{\alpha}\langle \psi_{\alpha}|\phi_i\rangle\langle \phi_i|\psi_{\alpha}\rangle.
\end{equation}
Therefore,
\begin{equation}
\begin{split}
 \tr{\rho(0)\, M} & =\sum_{\alpha}n_{\alpha}\langle \psi_{\alpha}|P|\psi_{\alpha}\rangle\\
 & = \sum_{\alpha}n_{\alpha}\tr{\psi_{\alpha}P},
 \end{split}
\end{equation}
where $P=\sum_i\ket{\phi_i}\bra{\phi_i}$.  For a state like $a^{\dagger}_1...a^{\dagger}_N\ket{0}$, the first $N$ occupation numbers $n_{\alpha}$ are one, so we get
\begin{equation}
 \tr{\rho(0)\, M} = \sum_{\alpha=1}^{N}\tr{\psi_{\alpha}P}.
\end{equation}
To \bl{account for the dependence on time in the formula}, we use $\bl{a(\ket{\psi_{\alpha}(t)})=}a_{\alpha}(t)=U(t)a_{\alpha}U^{\dagger}(t)$ and
\begin{equation}
 \tr{\rho(0)\, a^{\dagger}_{\alpha}a_{\beta}} = \tr{\rho(t)\, a^{\dagger}_{\alpha}(t)a_{\beta}(t)}.
\end{equation}
The end result is
\begin{equation}
 \tr{\rho(t) M}= \sum_{\alpha}n_{\alpha}\tr{\psi_{\alpha}(t) P}.
\end{equation}
\end{proof}

\section{Fluctuations}
\label{app:fluctuations}
It will be useful to prove the following lemma.
\begin{lemma}\label{lem:1}
Let $a^{\dagger}_{\alpha}=a^{\dagger}(\ket{\psi_{\alpha}})$ be a complete set of creation operators with $\{a_{\alpha},a^{\dagger}_{\beta}\}=\delta_{\alpha\beta}$.  And suppose we have a state $a^{\dagger}_{1}...a^{\dagger}_{N}\ket{0}$ with corresponding density operator $\rho$.  And take a measurement operator counting the number of particles in some modes $M=\sum_i b^{\dagger}_ib_i$, where $b_i=a(\ket{\phi_i})$. Then, we have that
\begin{equation}
\tr{\rho M^2}\leq\tr{\rho M}^2+\tr{\rho M}
\end{equation}
for fermions.  And for bosons, we have
\begin{equation}
\tr{\rho M^2}\leq\tr{\rho M}^2+\tr{\rho M}+N.
\end{equation}
\end{lemma}
\begin{proof}
First,
\begin{equation}
 \tr{\rho M^2} = \sum_{i,j}\tr{\rho (b^{\dagger}_i b_i)(b^{\dagger}_j b_j)}.
\end{equation}
Each term can be written as
\begin{equation}\label{eq:12}
 \tr{\rho (b^{\dagger}_i b_i)(b^{\dagger}_j b_j)} = \tr{\rho\, b^{\dagger}_j b^{\dagger}_i b_i b_j} + \delta_{ij}\tr{\rho\, b^{\dagger}_i b_i},
\end{equation}
which holds for bosons or fermions.  To make sense of this, we write
\begin{equation}
 b^{\dagger}_i=\sum_{\alpha}c_{i,\alpha}a^{\dagger}_{\alpha},
 \end{equation}
where $c_{i,\alpha}$ are complex numbers.  It will turn out below that only terms with $\alpha \in\{1,...,N\}=\mathcal{V}$ contribute.  Now we use the identity
\begin{equation}
 \bra{0}a_N...a_1(a^{\dagger}_{\alpha}a_{\beta})a^{\dagger}_1...a^{\dagger}_N\ket{0}=\delta_{\alpha\beta}.
\end{equation}
to get
\begin{equation}
 \tr{\rho\, b^{\dagger}_i b_i} =\sum_{\alpha,\beta\in\mathcal{V}}c_{i,\alpha}c^*_{i,\beta}\delta_{\alpha\beta}=\sum_{\alpha\in\mathcal{V}}|c_{i,\alpha}|^2.
\end{equation}

For fermions, we use the identity
\begin{equation}\label{eq:14}
 \bra{0}a_N...a_1(a^{\dagger}_{\alpha}a^{\dagger}_{\beta}a_{\gamma}a_{\varepsilon})a^{\dagger}_1...a^{\dagger}_N\ket{0}=[\delta_{\alpha\varepsilon}\delta_{\beta\gamma}-\delta_{\alpha\gamma}\delta_{\beta\varepsilon}]
\end{equation}
to get
\begin{equation}\label{eq:13}
\begin{split}
 \tr{\rho\, b^{\dagger}_j b^{\dagger}_i b_i b_j} & =\sum_{\alpha,\beta,\gamma,\varepsilon\in\mathcal{V}}c_{j,\alpha}c_{i,\beta}c^*_{i,\gamma}c^*_{j,\varepsilon}[\delta_{\alpha\varepsilon}\delta_{\beta\gamma}-\delta_{\alpha\gamma}\delta_{\beta\varepsilon}]\\
 & =\sum_{\alpha,\beta\in \mathcal{V}}[|c_{j,\alpha}|^2|c_{i,\beta}|^2-(c_{j,\alpha}c^*_{i,\alpha})(c_{i,\beta}c^*_{j,\beta})].
 \end{split}
\end{equation}
Now, using $c_{j,\alpha}=\{b_j^{\dagger},a_{\alpha}\}=\langle \psi_{\alpha}|\phi_j\rangle$, we get
\begin{equation}
 \begin{split}
  \sum_{\alpha\in\mathcal{V}}c_{j,\alpha}c^*_{i,\alpha}=\bra{\phi_i}Q\ket{\phi_j},\\
  \mathrm{where}\ Q=\sum_{\alpha\in\mathcal{V}}\ket{\psi_{\alpha}}\bra{\psi_{\alpha}}.
 \end{split}
\end{equation}
Then the second term in equation (\ref{eq:13}) becomes
\begin{equation}\label{eq:15}
 \sum_{\alpha,\beta\in\mathcal{V}}(c_{j,\alpha}c^*_{i,\alpha})(c_{i,\beta}c^*_{j,\beta})=|\bra{\phi_i}Q\ket{\phi_j}|^2,
\end{equation}
which is positive.  Therefore,
\begin{equation}
\begin{split}
 \tr{\rho\, b^{\dagger}_j b^{\dagger}_i b_i b_j} & \leq \sum_{\alpha,\beta\in\mathcal{V}}|c_{i,\alpha}|^2|c_{j,\beta}|^2\\
  & = \tr{\rho\, b^{\dagger}_i b_i}\tr{\rho\, b^{\dagger}_j b_j}.
 \end{split}
\end{equation}
Putting everything together gives
\begin{equation}
\tr{\rho M^2}\leq\tr{\rho M}^2+\tr{\rho M}.
\end{equation}

For bosons, the result is similar, but the identity in equation (\ref{eq:14}) is replaced by
\begin{equation}
\begin{split}
 & \bra{0}a_N...a_1(a^{\dagger}_{\alpha}a^{\dagger}_{\beta}a_{\gamma}a_{\varepsilon})a^{\dagger}_1...a^{\dagger}_N\ket{0}\\
 & =[\delta_{\alpha\varepsilon}\delta_{\beta\gamma}+\delta_{\alpha\gamma}\delta_{\beta\varepsilon}](1-\delta_{\alpha\beta}).
 \end{split}
\end{equation}
Because of this, following similar steps to those used to get equation (\ref{eq:13}), we get
\begin{equation}
 \tr{\rho\, b^{\dagger}_j b^{\dagger}_i b_i b_j} \leq\sum_{\alpha,\beta\in \mathcal{V}}|c_{j,\alpha}|^2|c_{i,\beta}|^2+|\bra{\phi_i}Q\ket{\phi_j}|^2.
\end{equation}
The extra term arising in our expression for $\tr{\rho M^2}$ is
\begin{equation}
\begin{split}
 \sum_{ij}|\bra{\phi_i}Q\ket{\phi_j}|^2 & =\sum_{ij}\bra{\phi_i}Q\ket{\phi_j}\bra{\phi_j}Q\ket{\phi_i}\\
 & =\tr{PQPQ}\leq N,
 \end{split}
\end{equation}
where $P=\sum_{i}\ket{\phi_i}\bra{\phi_i}$ and we used $\mathrm{rank}(Q)=N$.
\end{proof}

A corollary of this is that for pure states of bosons or fermions of the form $a^{\dagger}_{1}...a^{\dagger}_{N}\ket{0}$, the fluctuations satisfy
\begin{equation}
 \sigma_{M}^2=\tr{\rho M^2}-\tr{\rho M}^2 = O(N),
\end{equation}
where $N$ is the number of fermions or bosons in the system.  Furthermore, one can show that for $N$ bosons in the same mode, one also gets $\sigma_M^2 =O(N)$.

To say something more general about the fluctuations in fermion systems, we can also prove the following result.
\begin{theorem}\label{th:2}
 Given a non-interacting $N$ fermion system with corresponding single-particle Hamiltonian that has no degenerate energy levels and no degenerate energy gaps, then the time-average fluctuations satisfy
 \begin{equation}
  \left<\sigma_{M}(\rho(t))\right>\leq\sqrt{N},
 \end{equation}
 when the expectation value of $M$ on any infinite-time average state of $N$ particles is independent of the initial state.  Examples where this is true include the gases discussed in the examples in the main text and systems where $M$ counts the number of particles in a spatial region, provided the Hamiltonian is such that $\av{M}$, the time-average observable in the Heisenberg picture, is proportional to the total number operator.
\end{theorem}
\begin{proof}
Key to this result is the following inequality
\begin{equation}
\begin{split}
 \left<\sigma_{M}(\rho(t))\right> & \leq\sqrt{\left<\tr{\rho(t) M^2}-\tr{\rho(t) M}^2\right>}\\
 & = \sqrt{\tr{\av{\rho} M^2}-\left<\tr{\rho(t) M}^2\right>}\\
 & \leq \sqrt{\tr{\av{\rho} M^2}-\tr{\av{\rho} M}^2}\\
 & = \sigma_{M}(\av{\rho}),
 \end{split}
\end{equation}
where we used concavity of the square root in the first line and convexity of $f(x)=x^2$ to get to the second last line.

So it remains to calculate
\begin{equation}
 \sigma_{M}^2(\av{\rho})=\tr{\av{\rho} M^2}-\tr{\av{\rho} M}^2.
\end{equation}
As $\av{\rho}$ is the infinite-time average of $\rho(t)$, it follows that
\begin{equation}
 \av{\rho}=\sum_{E}p_E\, \omega_E,
\end{equation}
where $p_E$ is a normalized probability distribution and $\omega_E$ is a state with support only on the energy eigenspace corresponding to energy $E$.  Here $E$ labels the energy of the $N$ fermion system.  So $E$ is a sum of $N$ single-particle energies.  Let us write the creation operator that creates a fermion with energy $E_i$ as $a^{\dagger}_{i}$.  Now, the support of $\omega_E$ is spanned by the states
\begin{equation}
 a^{\dagger}_{i_1}...a^{\dagger}_{i_N}\ket{0},
\end{equation}
with $E_{i_1}+...+E_{i_N}=E$.  Therefore, given two different configurations with energy $E$, $\{i_1,...,i_N\}\neq\{j_1,...,j_N\}$,
\begin{equation}
  \bra{0}a_{i_N}...a_{i_1}(a^{\dagger}_{\alpha} a_{\beta})a^{\dagger}_{j_1}...a^{\dagger}_{j_N}\ket{0} = 0
\end{equation}
because single-particle energy levels are non-degenerate.  This implies
\begin{equation}
  \tr{\omega_E\, a^{\dagger}_{\alpha} a_{\beta}} = \sum_{\vec{i}\in C_E}q(\vec{i})\, \tr{P(\vec{i}) a^{\dagger}_{\alpha} a_{\beta}},
\end{equation}
where $\vec{i}$ is short for $\{i_1,...,i_N\}$, $C_E$ denotes all $\{i_1,...,i_N\}$ such that $E_{i_1}+...+E_{i_N}=E$, $q(\vec{i})$ is a normalized probability distribution such that $\sum_{\vec{i}\in C_E}q(\vec{i})=1$ and $P(\vec{i})$ is the projector onto $a^{\dagger}_{i_1}...a^{\dagger}_{i_N}\ket{0}$.  Similarly, it is a consequence of (single-particle) non-degenerate energy gaps that
\begin{equation}
  \tr{\omega_E\, a^{\dagger}_{\alpha}a^{\dagger}_{\beta}a_{\gamma}a_{\epsilon}} = \sum_{\vec{i}\in C_E}q(\vec{i})\, \tr{P(\vec{i})a^{\dagger}_{\alpha}a^{\dagger}_{\beta}a_{\gamma}a_{\epsilon}}.
\end{equation}
Now we can use lemma \ref{lem:1} to upper bound
\begin{equation}
\begin{split}
 & \tr{\av{\rho} M^2} = \sum_E\sum_{\vec{i}\in C_E}p_E\, q(\vec{i})\, \tr{P(\vec{i})M^2}\\
 & \leq \sum_E\sum_{\vec{i}\in C_E}p_E\, q(\vec{i})\left(\tr{P(\vec{i})M}^2+\tr{P(\vec{i})M}\right).
 \end{split}
\end{equation}
Next, we use $\tr{\omega M}=m$ independent of the state $\omega$ for fixed total particle number $N$, when $\omega$ is a time-average state, a special case of which is $\omega = P(\vec{i})$.  So
\begin{equation}
 \tr{\av{\rho} M^2}-\tr{\av{\rho} M}^2\leq m^2+m-m^2=m.
\end{equation}
Finally, using the fact that the expectation value of $M$ is bounded above by $N$ on the $N$ particle subspace leads to
\begin{equation}
 \sigma_{M}(\av{\rho})^2\leq N
\end{equation}
and therefore
\begin{equation}
\langle\sigma_{M}(\rho(t))\rangle\leq \sqrt{N}.
\end{equation}
\end{proof}

\section{Calculations for Fermions in a Box}
\label{app:Calculations for Fermions in a Box}
\bl{Equation (\ref{eq:41}) gives the distinguishability at time $t$,
\begin{equation}\label{eq:1078}
\begin{split}
 D_{P}(\sigma(t),\langle\sigma\rangle) = \abs{\sum_{n\neq m} \cos\left[(n^2-m^2)\nu t\right] \sigma_{nm}P_{mn}}.
 \end{split}
\end{equation}
And from equation (\ref{eq:40}), we have}
\begin{equation}
 \sigma_{nm}P_{mn}=\frac{1}{N}\sum_{k=1}^N\langle n| \psi_k\rangle\langle \psi_k| m\rangle \langle m|P| n\rangle.
\end{equation}
\bl{Evaluating} the inner products\bl{, we} get
\begin{equation}
 \begin{split}
 \langle m|P| n\rangle & = f(m,n)\\
  \langle n| \psi_k\rangle & = \sqrt{2}f(n,2k),
 \end{split}
\end{equation}
where
\begin{equation}
 f(n,m)=\begin{cases}
         \frac{1}{\pi}\left[\frac{\sin\left[(n-m)\frac{\pi}{2}\right]}{n-m}-\frac{\sin\left[(n+m)\frac{\pi}{2}\right]}{n+m}\right]\ \mathrm{if}\ n\neq m\\
         \frac{1}{2}\ \mathrm{if}\ n=m.
        \end{cases}
\end{equation}
Notice that, if both $x$ and $y$ are even or odd, unless $x=y$, then $f(x,y)=0$.  The net result of this is
\begin{equation}
 \sigma_{nm}P_{mn}=\frac{1}{N}\sum_{k=1}^Nf(n,2k)^2
\end{equation}
if $m=2k$ and $n$ is odd and similarly if $n=2k$ and $m$ is odd.  All other terms are zero.  Then subbing this into equation (\bl{\ref{eq:1078}}), the distinguishability $D_{P}(\sigma(t),\langle\sigma\rangle)$ becomes
\begin{equation}
 \frac{2}{N}\abs{\sum_{n\ \mathrm{odd}}\sum_{k=1}^N\cos\left[(n^2-4k^2)\nu t\right]f(n,2k)^2}.
\end{equation}
Furthermore, using $\sin(r\pi/2)=({-1})^{(r-1)/2}$, which holds for odd $r$, we get
\begin{equation}
 f(n,2k)^2=\frac{4}{\pi^2}\frac{4k^2}{(n^2-4k^2)^2},
\end{equation}
for $n\neq 2k$.

Next, we can \bl{find} a bound for the equilibration time.  First, we make the substitution $n=2k+l$, noting that $l> -2k$ since $n>0$, and $l$ must be odd.  It follows that $f(2k+l,2k)^2$ is peaked around small values of $l$, so we can focus on terms with $l\in \mathcal{S}$, where $\mathcal{S}$ contains all odd integers from $-K$ to $K$.  To quantify the resulting error, we use
\begin{equation}
 f(2k+l,2k)^2=\frac{4}{\pi^2}\frac{4k^2}{l^2(4k+l)^2}\leq \frac{4}{\pi^2l^2}.
\end{equation}
The sum of all terms with $l\notin \mathcal{S}$ can be bounded above by
\begin{equation}
\begin{split}
& \sum_{\mathrm{odd}\ l\notin \mathcal{S}}f(2k+l,2k)^2\leq \sum_{\mathrm{odd}\ l\notin \mathcal{S}}\frac{4}{\pi^2l^2}\\
 & \leq\sum_{l=-2k+1}^{-K-2}\frac{4}{\pi^2l^2}+\sum_{l=K+2}^{\infty}\frac{4}{\pi^2l^2}\\
 & \leq \frac{8}{\pi^2}\sum_{l=K+2}^{\infty}\frac{1}{l^2}\leq \frac{8}{\pi^2}\sum_{l=K+2}^{\infty}\left(\frac{1}{l-1}-\frac{1}{l}\right)\\
 & = \frac{8}{\pi^2}\frac{1}{K+1}.
 \end{split}
\end{equation}
So neglecting terms corresponding to $l\notin \mathcal{S}$ results in an upper bound for $D_{P}(\sigma(t),\langle\sigma\rangle)$ of
\begin{equation}
 \frac{2}{N}\abs{\sum_{l\in \mathcal{S}}\sum_{k=1}^N\cos\left[(4kl+l^2)\nu t\right]\frac{4}{\pi^2}\frac{4k^2}{l^2(4k+l)^2}}+\frac{16}{\pi^2}\frac{1}{K+1},
\end{equation}
which follows from the triangle inequality and $|\cos(x)|\leq 1$.  Next, as $f(2k+l,2k)$ is awkward to work with, we use
\begin{equation}
\begin{split}
 \abs{\frac{1}{\pi^2l^2}-f(2k+l,2k)^2} & =\frac{1}{\pi^2\abs{l}}\abs{\frac{8k+l}{(4k+l)^2}}\\ 
 & \leq \frac{3}{2\pi^2}\frac{1}{k|l|},
 \end{split}
\end{equation}
where we used the triangle inequality and the fact that $1/(4k+l)< 1/(2k)$, since $l>-2k$.

This allows us to upper bound $D_{P}(\sigma(t),\langle\sigma\rangle)$ by
\begin{equation}
 \frac{2}{N}\abs{\sum_{l\in \mathcal{S}}\sum_{k=1}^N\cos\left[(4kl+l^2)\nu t\right]\frac{1}{\pi^2l^2}}+\mu,
\end{equation}
where
\begin{equation}
  \mu = \frac{16}{\pi^2}\frac{1}{K+1}+\frac{6}{\pi^2}\frac{\left[\ln(N)+1\right]\left[\ln(K)+1\right]}{N}.
\end{equation}
To get this, we used the triangle inequality and the inequality $\sum_{r=1}^R 1/r<\ln(R)+1$.

Using the triangle inequality again, we get
\begin{equation}\label{eq:10}
 \begin{split}
 & D_{P}(\sigma(t),\langle\sigma\rangle)\\
 & \leq \frac{2}{N}\sum_{l\in\mathcal{S}}\frac{1}{\pi^2l^2}\abs{\sum_{k=1}^N\cos\left[(4kl+l^2)\nu t\right]}+\mu\\
 & =2\sum_{l\in\mathcal{S}}\abs{\frac{\sin[2Nl\nu t]\cos[(2(N+1)+l)l\nu t]}{N\pi^2l^2\sin[2l\nu t]}}+\mu\\
 & \leq 2\sum_{l\in\mathcal{S}}\abs{\frac{\sin[2Nl\nu t]}{N\pi^2l^2\sin[2l\nu t]}}+\mu.
 \end{split} 
\end{equation}
In the third line we used the identity \cite{Knapp}
\begin{equation}
 \sum_{k=0}^{N-1}\cos(\alpha k+\phi)=\frac{\sin[N\alpha/2]\cos[(N-1)\alpha/2 +\phi]}{\sin(\alpha/2)}.
\end{equation}
Let us look at each term in the sum in the last line of equation (\ref{eq:10}) separately.  They have period $\f{\pi}{2l\nu}$, so we need only focus on this interval to find the time average of $D_{P}(\sigma(t),\langle\sigma\rangle)$.

When $2l\nu t$ is close to $0$ or $\pi$, the $\sin[2l\nu t]$ in the denominator in the last line in equation (\ref{eq:10}) is small.  So for $t$ such that $2l\nu t\in[0,\f{1}{Na})$ or $2l\nu t\in(\pi -\f{1}{Na},\pi]$, where $a$ is a small constant we will choose at the end, we bound
\begin{equation}
 \abs{\frac{\sin[2Nl\nu t]}{N\pi^2l^2\sin[2l\nu t]}}\leq \frac{1}{\pi^2 l^2}.
\end{equation}
When $2l\nu t\in[\f{1}{Na},\f{\pi}{2}]$, we can use the inequality $\sin(x)\geq 2x/\pi$ for all $x\in[0,\pi/2]$ to get
\begin{equation}\label{eq:11}
 \abs{\frac{\sin[2Nl\nu t]}{N\pi^2l^2\sin[2l\nu t]}}\leq \frac{1}{N\pi |l|^3 4\nu t}.
\end{equation}

\begin{figure}[ht!]
\centering
    \resizebox{8.0cm}{!}{\includegraphics{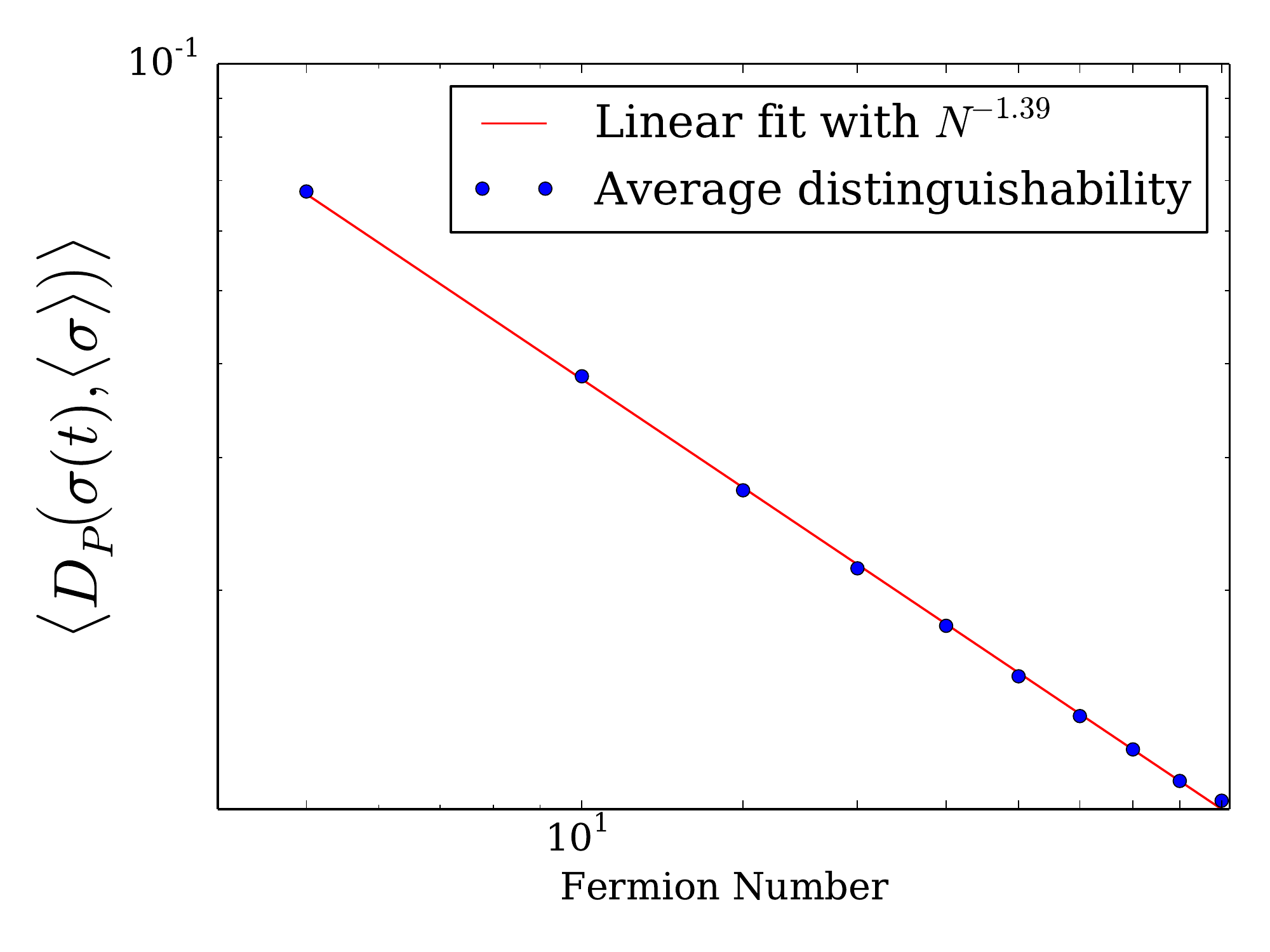}}
    \footnotesize{\caption[Plot of equilibration]{This plot shows the time-average distinguishability $\langle D_{P}(\sigma(t),\langle\sigma\rangle)\rangle$ as a function of particle number from $5$ to $50$ fermions on a log scale.  Numerically, $\langle D_{P}(\sigma(t),\langle\sigma\rangle)\rangle =O(N^{-1.39})$, which is faster than the bound in equation (\ref{eq:16}), which was $O(\ln(N)^2/N)$.}
    \label{fig:fermions_in_a_box2}}
\end{figure}

To find the time average of $D_{P}(\sigma(t),\langle\sigma\rangle)$, we use the fact that the average over $[\f{\pi}{4l\nu},\f{\pi}{2l\nu}]$ is the same as that over $[0,\f{\pi}{4l\nu}]$ by symmetry.  So we need only average each term over $[0,\f{\pi}{4l\nu}]$.  The result is
\begin{equation}
\begin{split}
 & \left< \abs{\frac{\sin[2Nl\nu t]}{N\pi^2l^2\sin[2l\nu t]}}\right> \\ & \leq \frac{1}{\pi^2 l^2}\frac{2}{\pi Na}
 + \frac{4l\nu}{\pi}\left(\int_{1/(2Nal\nu)}^{\pi/(4l\nu)}\frac{\mathrm{d}t}{N\pi |l|^3 4\nu t}\right)\\
 & = \frac{2}{\pi^3 l^2 N a} + \frac{\ln(\pi N a/2)}{N\pi^2 l^2}.
 \end{split}
\end{equation}
Plugging this into equation (\ref{eq:10}), we get
\begin{equation}
\left< D_{P}(\sigma(t),\langle\sigma\rangle)\right> \leq \frac{4}{3\pi N a} + \frac{2\ln(\pi N a/2)}{3N} + \mu,
\end{equation}
where we used $\sum_{l\in \mathcal{S}}1/l^2\leq 2\sum_{l=1}^{\infty}1/l^2=\pi^2/3$.  If we choose $K=N$, then we see that
\begin{equation}\label{eq:16}
 \left< D_{P}(\sigma(t),\langle\sigma\rangle)\right>=O\left(\frac{\ln(N)^2}{N}\right).
\end{equation}
So for large $N$, this is extremely small and equilibration occurs.  In fact, as figure \ref{fig:fermions_in_a_box2} shows, the time-average distinguishability decays faster with $N$ than the bound here.

In figure \ref{fig:fermions} we saw that $D_{P}(\sigma(t),\langle\sigma\rangle)$ becomes small and then stays small for most times.  In order to find the equilibration time, we can upper bound the time it takes for the distinguishability to become small.  Plugging $t=\frac{1}{2Na\nu}$ into the bound in equation (\ref{eq:11}), gives the bound
\begin{equation}
 D_{P}(\sigma(t),\langle\sigma\rangle)  \leq \sum_{l\in\mathcal{S}}\frac{a}{\pi |l|^3 }+\mu
  \leq \frac{\pi a}{3}+\mu.
\end{equation}
Here we choose $a$ to be a small constant such that the distinguishability is small at $t=\frac{1}{2Na\nu}$.  Then the equilibration time is bounded above by
\begin{equation}
 T_{\mathrm{eq}}=\frac{1}{2Na\nu}.
\end{equation}

\bl{
\subsection{Three dimensions}}
We can extend this to a three dimensional example in a way similar to the extension to three dimensions for a particle in a box in \cite{MLS15}.  Suppose the initial state of the $N$ fermion system is
\begin{equation}
 \prod_{\vec{j}\in J}a^{\dagger}(\ket{\psi_{\vec{j}}})\ket{0},
\end{equation}
where $J$ is the set of three-component vectors with components in $\{1,...,J_{\max}\}$, so we have $N=J_{\max}^3$.  And let $\ket{\psi_{\vec{j}}}$ be the energy eigenstate for a particle in the left half of a box labelled by $\vec{j}$ analogous to $\ket{\psi_k}$ in one dimension in equation (\ref{eq:42}).  Suppose the observable we are considering is that which counts the number of particles in the left half of the box.

After mapping to the single-particle picture, the distinguishability becomes
\begin{equation}
 D_{P}(\sigma(t),\avg{\sigma})=D_{P_x}(\sigma_x(t),\avg{\sigma_x}),
\end{equation}
where $P=P_x\otimes \openone_y \otimes \openone_z$ is the projector onto the left half of the box.  Here $\sigma_x(t)$ is the reduced state of the system on $\mathcal{H}_x$, the Hilbert space corresponding to the $x$ degrees of freedom.  We have
\begin{equation}
 \sigma_x(0)=\frac{1}{J_{\max}}\sum_{j=1}^{J_{\max}}\ket{\psi_{j}}\bra{\psi_{j}},
\end{equation}
where $\ket{\psi_{j}}$ is the $j$th energy eigenstate of a particle trapped in the left hand side of a one dimensional box.  So the problem is now equivalent to the one dimensional example.  Therefore, the equilibration timescale is at most $T_{\mathrm{eq}}=O(1/J_{\max})= O(1/N^{1/3})$.\newline

\bl{
\section{Single-particle equilibration}
Here we will derive a useful formula that shows when equilibration occurs.  The proof is very similar to one in \cite{MGLFS14}, mainly differing by using a different weight for the time average.
\begin{lemma}
\label{th:exptval}
Take a finite dimensional system evolving via a time independent Hamiltonian in the state $\sigma(t)$.  For any
operator $A$ and time $T >0 $
\begin{equation}
\label{eq:eqofexpct}
\begin{split}
 \frac{\av{ |\tr{\sigma(t)A} - \tr{ \av{\sigma} A }|^2 }_T}{\norm{A}^2} \leq
 \frac{c_1}{\de}\max_{\beta}\sum_{\alpha}e^{-(G_{\alpha}-G_{\beta})^2\f{T^2}{16}},
\end{split}
\end{equation}
where $c_1=e\f{\sqrt{\pi}}{2}$ and $G_{\alpha}=E_i-E_j$ denote the non-zero energy gaps, so $E_i\neq E_j$.  Also,
\begin{equation}
 \frac{1}{d_{\mathrm{eff}}}=\sum_{E}\left(\tr{\sigma(0) P_{E}}\right)^2,
\end{equation}
where $P_E$ is the projector onto the energy eigenspace corresponding to energy $E$.
\end{lemma}
\begin{proof}
First, we will take $\sigma(t)$ to be pure, extending the result to mixed states at the end.  Because $\sigma(t)=\ket{\psi(t)}\bra{\psi(t)}$ is pure, we can choose an eigenbasis of $H$ where $\ket{\psi(t)}$ only overlaps with a single eigenstate $\ket{n}$ for each energy level.  This means that degenerate energy levels will not cause any problems.  The state at time $t$ is
\begin{equation}
\ket{\psi(t)} = \sum_{n} c_n e^{-i E_n t} \ket{n},
\end{equation}
where $c_n = \braket{n}{ \psi(0)}$.  The time-average state is $\av{\sigma} = \sum_n |c_n|^2 \proj{n}$, and the effective dimension is given by $1/\de = \sum_n |c_n|^4 = \tr{\av{\sigma}^2}$.

Using the notation $A_{ij}= \bra{i} A \ket{j}$, we have
\begin{align}
\label{eq:eqexpt1}
\langle |&\tr{\sigma(t)A}  -\tr{\av{\sigma} A}|^2 \rangle_T\\
& = \Big\langle \Big|\sum_{i \neq j}  (c_j^*A_{ji} c_i) e^{-i (E_i-E_j)t} \Big|^2\Big\rangle_T \nonumber \\
 & \leq  \sum_{\scriptsize \begin{array}{c} i \neq j \\ k\neq l \end{array}} \!\!\! \! (c_j^* A_{ji} c_i)  (c_l^* A_{lk} c_k)^* \av{e^{i[(E_k-E_l) - (E_i - E_j)]t} }_T.
\end{align}
To make our expressions more concise, we denote non-zero energy gaps by $G_{\beta}=E_i-E_j$, with $\beta = (i,j)$, where $i\neq j$.  We also define the vector
\begin{equation}
v_{\beta}=v_{(i,j)} =  c_j^* A_{ji} c_i
\end{equation}
and the Hermitian matrix
 \begin{equation}
 M_{\alpha\beta} = \av{ e^{i(G_{\alpha} - G_{\beta})t} }_T.
\end{equation}
Equation (\ref{eq:eqexpt1}) becomes
\begin{equation}
\begin{split}
\av{ |\tr{\sigma(t)A} - \tr{\av{\sigma} A}|^2 }_T &=  \sum_{\alpha, \beta} v_{\alpha}^* M_{\alpha\beta}  v_{\beta}\\
 & \leq  \norm{M} \sum_{\alpha}  |v_{\alpha}|^2   \\
 & =\norm{M} \sum_{i\neq j}  |c_i|^2 |c_j|^2 |A_{ji}|^2  \\
 & \leq \norm{M} \sum_{i, j}  |c_i|^2 |c_j|^2 |A_{ji}|^2 \\
 & = \norm{M} \tr{A \av{\sigma} A^{\dag} \av{\sigma} }.
 \end{split}
 \end{equation}
 As $\tr{A^{\dag} B}$ defines an inner product for operators, we may use the Cauchy-Schwartz inequality.  We can also use the inequality $\tr{PQ} \leq \|P\| \tr{Q} $, which holds for $P$ and $Q$ positive operators.  Then we get
\begin{equation}
 \begin{split}
 &\av{ |\tr{\sigma(t)A} - \tr{\av{\sigma} A}|^2 }_T\\
 &\leq  \norm{M}\sqrt{\tr{A^{\dag}\!A\, \av{\sigma}^2} \tr {A A^{\dag} \av{\sigma}^2}} \\
 &\leq \norm{M}\norm{A}^2 \tr{\av{\sigma}^2} \\
&=\frac{\norm{M}\norm{A}^2}{\de}.
 \end{split}
\end{equation}
Next, we can use the inequality for matrix norms \cite{Horn85}
\begin{equation}
\label{eqn:norm}
\begin{split}
\norm{M} & \leq \sqrt{\vertiii{M}_{\infty}\vertiii{M}_{1}}\\
&= \max_{\beta} \sum_{\alpha} |  M_{\alpha\beta}|,
\end{split}
\end{equation}
where $\vertiii{M}_{1}$ and $\vertiii{M}_{\infty}$ are the column and row matrix norms respectively.  The second line holds because $M$ is hermitian, implying $\vertiii{M}_{\infty}=\vertiii{M}_{1}$.

Our next task is to deal with $M_{\alpha\beta}= \av{ e^{i(G_{\alpha} - G_{\beta})t} }_T$.  This can be simplified by replacing the original time average over the interval $[0,T]$ by a differently weighted average \cite{MGLFS14,FWBSE14}.  This works because the quantity we are averaging (see equation (\ref{eq:eqexpt1})) is positive, and because the new weight $f(t)$ satisfies $f(t)\geq 1/T$ on the interval $[0,T]$.

We will choose the weight to be a Gaussian.  Then for any positive $g(t)$, we have
\begin{equation}
\begin{split}
  \frac{1}{T}\int_0^T\! \textrm{d}t\, g(t) & \leq \frac{e}{T}\int_0^T\! \textrm{d}t\, g(t)\, e^{-4(t-T/2)^2/T^2} \\
 & \leq \frac{e}{T}\int_{{-}\infty}^{\infty}\! \textrm{d}t\, g(t)\, e^{-4(t-T/2)^2/T^2}.
 \end{split}
\end{equation}
Here, $e=\exp(1)$.

Employing this weighted averaging, the matrix elements of M are
\begin{equation}
\label{eqM}
M_{\alpha\beta} = \frac{e}{T}\int_{{-}\infty}^{\infty}\!\textrm{d}t\, e^{-4(t-T/2)^2/T^2}\, e^{i(G_{\alpha} - G_{\beta})t}.
\end{equation}
So we get
\begin{equation}
\label{eqM1}
|M_{\alpha\beta}|= c_1\,e^{-(G_{\alpha} - G_{\beta})^2T^2/16},
\end{equation}
where $c_1=e\sqrt{\pi}/2$.  From equation (\ref{eqn:norm}), we get
\begin{equation}
\norm{M} \leq c_1\max_{\beta}\sum_{\alpha} e^{-(G_{\alpha} - G_{\beta})^2T^2/16}.
\end{equation}
Finally, we arrive at
\begin{equation}
\begin{split}
 \frac{\av{ |\tr{\sigma(t)A} - \tr{ \av{\sigma} A }|^2 }_T}{\norm{A}^2} \leq
 \frac{c_1}{\de}\max_{\beta}\sum_{\alpha} e^{-(G_{\alpha} - G_{\beta})^2\f{T^2}{16}}.
\end{split}
\end{equation}

The final step of the proof is to extend the result to mixed states by doing a purification \cite{NC00}, as in \cite{Short10}.  Denote the system's Hilbert space by $\mathcal{H}_S$.  Then we can define a pure state $\ket{\psi(0)}$ on $\mathcal{H}_S \otimes \mathcal{H}_A$, with $\mathrm{dim}(\mathcal{H}_S)=\mathrm{dim}(\mathcal{H}_A)$, with the property that the reduced state on the first system is $\mathrm{tr}_{A}\left[\ket{\psi(0)}\bra{\psi(0)}\right]=\sigma(0)$.  We recover the original evolution $\sigma(t)$ of the first system by evolving $\ket{\psi(t)}$ under the joint Hamiltonian $H \otimes \openone$.  Crucially, the expectation value of any operator $A$ on the state $\sigma(t)$ is the same as the expectation value of $A \otimes \openone$ on the purified state on the joint system.  Also, $\|A\|=\|A\otimes \openone\|$.  Finally, the effective dimension of the purified system is equal to the effective dimension of the original system, which can be seen from from $\mathrm{tr}[P_E\sigma(0)]=\textrm{tr}[P_E\otimes \openone \proj{\psi(0)}]$.
\end{proof}

Our remaining task is to simplify things in terms of the density of states.  In the sum over $\alpha$, we can separate out the time dependent term, which has $G_{\alpha}\neq G_{\beta}$, and evaluate the sum by making the density of states approximation.  We make the replacement $\sum_{E}= \int \!\mathrm{d}{E}\,n(E)$, where $n(E)$ is the density of states, so that
\begin{equation}
\begin{split}
& \max_{\beta} \sum_{\substack{\alpha\\G_{\alpha}\neq G_{\beta}}} e^{-(G_{\alpha} - G_{\beta})^2\f{T^2}{16}}\\
= & \max_{\beta}\!\!\!\sum_{\substack{E^{\prime}\\E^{\prime}\neq E\\E-E^{\prime}\neq G_{\beta}}}\!\!\!\sum_{E} e^{-(E-E^{\prime}-G_{\beta})^2\f{T^2}{16}}\\
= & \max_{\beta}\int_{E_{\min}}^{E_{\max}}\!\mathrm{d}E^{\prime}\,n(E^{\prime})\int_{E_{\min}}^{E_{\max}}\!\mathrm{d}E\,n(E)\, e^{-(E-E^{\prime}-G_{\beta})^2\f{T^2}{16}}\\
\leq &\, n_{\max}\max_{\beta}\int_{E_{\min}}^{E_{\max}}\!\mathrm{d}E^{\prime}\,n(E^{\prime})\int_{E_{\min}}^{E_{\max}}\!\mathrm{d}E \,e^{-(E-E^{\prime}-G_{\beta})^2\f{T^2}{16}}\\
\leq &\, n_{\max} \int_{E_{\min}}^{E_{\max}}\!\mathrm{d}E^{\prime}\,n(E^{\prime})\frac{4\sqrt{\pi}}{T}\\
= &\, n_{\max} \frac{4\sqrt{\pi}d}{T},
\end{split}
\end{equation}
where $n_{\max}=\max_En(E)$ and $d$ is the dimension of the particle's state space.  We also used $\int \!\mathrm{d}{E}\,n(E)=d$ to get the last line.

We define $D_G$ to be the maximum number of gaps $G_{\alpha}$ of the same size.  In other words,
\begin{equation}
D_G=\max_{\beta} \sum_{\substack{\alpha\\G_{\alpha}= G_{\beta}}}1.
\end{equation}
So finally we get
\begin{equation}
\begin{split}
\frac{\av{ |\tr{\sigma(t)A} - \tr{ \av{\sigma} A }|^2 }_T}{\norm{A}^2} \leq
\frac{c_1}{\de}\left[D_G+n_{\max} \frac{4\sqrt{\pi}d}{T}\right].
\end{split}
\end{equation}
It is good to point out that the density of states approximation misses degenerate energy gaps because they are measure zero.

\section{Free lattice models}\label{sec:Free lattice models}
We want to see when equilibration occurs for free lattice models, and also estimate the timescale.
A key step is to prove equilibration with respect to single-mode measurements.  A consequence is corollary \ref{cor:143}, which implies equilibration occurs for any measurement on $K$ local modes, provided $K$ is relatively small, and the initial state is Gaussian and commutes with the total number operator.

Consider the observable $M=a^{\dagger}(\ket{\phi})a(\ket{\phi})$, which counts the number of particles in the state $\ket{\phi}$.  By applying equation (\ref{eq:18}), we have
\begin{equation}
\tr{\rho(t) M}= \sum_{j=1}^{N}n_j\tr{\psi_j(t) \ket{\phi}\bra{\phi}},
\end{equation}
The trick now is to switch to the Heisenberg picture.  Then
\begin{equation}
\tr{\rho(t) M}= \sum_{j=1}^{N}n_j\tr{\psi_j \ket{\phi({-t})}\bra{\phi({-}t)}}.
\end{equation}
Because of this, we can think of $\ket{\phi({-t})}\bra{\phi({-}t)}=\sigma(-t)$ as the state of a particle.

For any fermionic state, or a bosonic state with at most one boson in $N$ orthogonal modes, we have
\begin{equation}
 \Pi=\sum_{j=1}^{N}n_j\psi_j \leq \openone.
\end{equation}
So we can think of this as a measurement (POVM) operator.  For a system of bosons with more than one boson in each mode, the result can be extended simply by factoring out the maximum number of bosons in a mode, but we will not include this in the following formulas.  Then we have
\begin{equation}\label{eq:102}
\begin{split}
 \abs{\tr{\rho(t) M}-\tr{\avg{\rho} M}} & = \abs{\tr{\Pi\sigma(-t)}-\tr{\Pi\avg{\sigma}}}\\
 & = D_{\Pi}(\sigma(-t),\avg{\sigma}).
 \end{split}
\end{equation}
Applying equation (\ref{eq:5}) in theorem \ref{th:54}, we get
\begin{equation}\label{eq:98}
\begin{split}
\av{D_{\Pi}(\sigma(-t),\avg{\sigma})}_T\leq\sqrt{\frac{c_1}{\de}\left[D_G + c_2\frac{n_{\max}d}{T}\right]}.
 \end{split}
\end{equation}
The first task is to estimate $1/\de$.  Because the measurement operator $M=a^{\dagger}(\ket{\phi(0)})a(\ket{\phi(0)})$ is local, the state $\ket{\phi(0)}\bra{\phi(0)}$ is localized.  But for free lattice models, the energy eigenstates are spread out over the whole lattice:\ they often have the form $\ket{\phi(\vec{p})}\ket{\vec{p}}$, where $\ket{\vec{p}}$ is the momentum state with momentum $\vec{p}$, and $\ket{\phi(\vec{p})}$ is the state of the extra degree of freedom (spin or particle type, for example).  In that case, given an energy eigenstate $\ket{E}$ we get $|\langle E\ket{\phi(0)}|^2\leq l/V$, where the factor of $l$ appears because $\ket{\phi(0)}$ is spread over at most $l$ lattice sites.  This implies
\begin{equation}
 \frac{1}{d_{\mathrm{eff}}}=\sum_{E}\left(\tr{\sigma(0) P_{E}}\right)^2\leq \frac{n_d^2l^2d}{V^2},
\end{equation}
where $d$ is the dimension of the Hilbert space and $n_d$ is the maximum degeneracy of the energy levels.  We also have that $d=s V$, where $s$ is the number of orthogonal states at each site.  For a spin $1/2$ particle, we would have $s=2$.  Plugging this into equation (\ref{eq:98}), we get
\begin{equation}
\begin{split}
\av{D_{\Pi}(\sigma(-t),\avg{\sigma})}_T\leq l\sqrt{\frac{c_3}{d}\left[D_G + c_2\frac{n_{\max}d}{T}\right]},
 \end{split}
\end{equation}
where $c_3=c_1 n_d^2s^2$.  

The last task is to estimate $n_{\max}=\max_En(E)$.

\subsection{Density of states for lattice models}\label{eq:Density of states for lattice models}
For lattice models, estimating $\max_En(E)$ causes problems because $n(E)$ often diverges.  Fortunately, we can truncate to a slightly smaller energy range, such that $\max_En(E)$ is bounded, and the error caused by the truncation is small.

Let us take a nearest-neighbour hopping model on a line as an example.  The corresponding single-particle Hamiltonian is
\begin{equation}
 H=\frac{1}{2}\sum_{i=1}^{L}\big(\ket{i}\bra{i+1}+\ket{i+1}\bra{i}\big),
\end{equation}
where we are assuming translational invariance, so the site at $L+1$ is identified with site $1$.  Switching to momentum space diagonalizes this, and we get the dispersion relation $E(p)=\cos(p)$, where $p=2\pi k/L$, with $k\in \{1,...,L\}$.  (In making the density of states approximation, we assume that $L$ is large but finite.)  The density of states is
\begin{equation}
\begin{split}
 n(E) & = \frac{L}{\pi}\frac{\textrm{d}p}{\textrm{d}E} = \frac{L}{\pi}\frac{1}{\abs{\sin\left(\arccos(E)\right)}}\\
 & =\frac{L}{\pi}\frac{1}{\sqrt{1-E^2}}.
 \end{split}
\end{equation}
\begin{figure}[ht!]
 \resizebox{7.5cm}{!}{\includegraphics{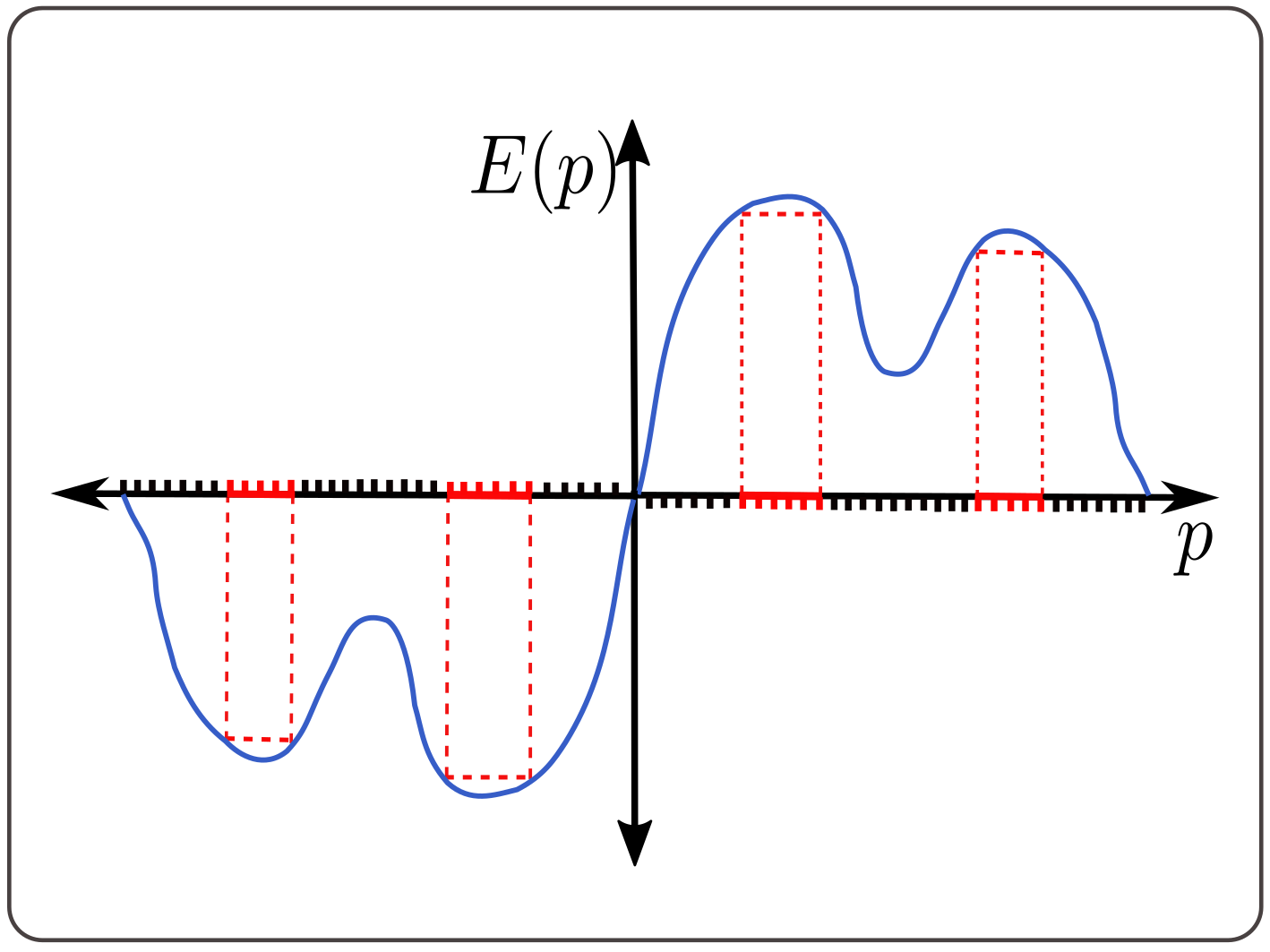}}
    \footnotesize{\caption[Dispersion Relation]{\bl{For a given dispersion relation of a lattice model, we need to truncate the state to avoid points where the density of states diverges.  These points correspond to turning points in the dispersion relation.  The picture shows an exaggerated version of the regions we would exclude for a one dimensional example, which are the red sections on the $p$ axis.  The number of states in a momentum window of fixed width scales linearly with $L$ for one dimensional systems.  More generally, the number of states in such a region scales linearly with $V$, the number of sites.}\label{fig:dispersion}}}
\end{figure}
At $E=\pm 1$ this is infinite.  However, we can truncate the state, neglecting all terms with $E\in[-1,-\cos(p_0))\cup(\cos(p_0),1]$, where we choose a fixed $p_0=2\pi k_0/L>0$ to be small.  This leads to a constant error in approximation of the state as there is a constant fraction of energy eigenstates with energy in this range.  Defining $P_0$ to be the projector onto the subspace corresponding to the omitted energy range, we get
\begin{equation}
\begin{split}
 \big(\|\ket{\phi}-(\openone-P_0)\ket{\phi}\|_2\big)^2 & = 1-\bra{\phi}(\openone-P_0)\ket{\phi}\\
 &= \bra{\phi}P_0\ket{\phi}\\
 &\leq N_0\frac{l}{L},
 \end{split}
\end{equation}
where we used $|\langle E\ket{\phi(0)}|^2\leq l/V$ from the previous section to get the last line.  And $N_0$ is the number of states with energy in the excluded set.  Next, we can use that $N_0=4k_0=2p_0L/\pi$ to get
\begin{equation}
 \big(\|\ket{\phi}-(\openone-P_0)\ket{\phi}\|_2\big)^2\leq \frac{2p_0l}{\pi}.
\end{equation}
So the error in approximating $\ket{\phi}$ by a state restricted to the smaller energy range can be made small by choosing $p_0$ to be small.  Furthermore, we have that
\begin{equation}
 \max_E n(E)= \frac{L}{\pi}\frac{1}{\abs{\sin(p_0)}}.
\end{equation}
Crucially, $p_0$ is fixed and does not depend on the number of sites.

The same ideas apply to other dispersion relations (for example, those arising from translationally invariant Hamiltonians, possibly in higher spatial dimensions).  See figure \ref{fig:dispersion}.  The basic idea is to exclude regions where the density of states diverges, which corresponds to points where the dispersion relation is flat.  The key point is that the density of states is a fixed function, and the fraction of energy eigenstates corresponding to regions where it is large remains constant.  (Notice that the trivial Hamiltonian $H=\mathrm{constant}$ has a completely flat dispersion relation, so that this trick will not work in that case.)  But generally for free lattice models we expect $\max_E n(E)\propto V$, where $V$ is the number of lattice sites.

\subsection{From single-mode to multi-mode measurements}\label{sec:From single-mode to multi-mode measurements}
It is possible to relate $K$ mode measurements to single-mode measurements if the state is Gaussian.  We will only prove this here for fermionic Gaussian states, as the bosonic analogue is similar.
\begin{repcorollary}{cor:143}
 Suppose $\rho=\rho(0)$ is Gaussian and satisfies $[\rho,N]=0$, where $N$ is the total number operator.  (This is still quite general, but it rules out BCS states.)  Let $M$ be a measurement operator acting on $K$ modes on $l$ sites, where the modes are local, but not necessarily near each other.  Then we have
 \begin{equation}
 \begin{split}
  & \avg{\abs{\tr{\rho M(t)}-\tr{\rho \av{M}}}}_T\\
  & \leq 2^{K+2}mlK\sqrt{c_3\left[\frac{D_G}{d}+ c_2\frac{n_{\max}}{T}\right]},
  \end{split}
 \end{equation}
with $m= \max|m_{r_1,...,r_{2K}}|$, where $m_{r_1,...,r_{2K}}$ are the coefficients of $M$ when expanded in a fermionic operator basis on the $K$ modes.  See equation (\ref{eq:64}).  And $c_2=4\sqrt{\pi}$ and $c_3=c_1 n_d^2s^2$, where $c_1=e\sqrt{\pi}/2$.
\end{repcorollary}
\begin{proof}
Define $c_{\alpha}$, with $\alpha\in\{1,...,2K\}$, to be $2K$ Majorana operators generating the algebra for the $K$ modes that $M$ acts on.  We can choose these $c_{\alpha}$ to make the covariance matrix simple, as we will see below.
We will work in the Heisenberg picture here, but suppress the time dependence and just write $c_{\alpha}$ instead of $c_{\alpha}(t)$.  We can expand $M$ in terms of these Majorana operators:
 \begin{equation}\label{eq:64}
 \begin{split}
  M =\!\!\sum_{r_1,...,r_{2K}=0,1}\!\!m_{r_1,...,r_{2K}}c^{r_1}_{1}...c^{r_{2K}}_{2K}.
  \end{split}
 \end{equation}
 Using the triangle inequality, we get
 \begin{equation}
 \begin{split}
  & \abs{\tr{\rho M(t)}-\tr{\rho \av{M}}}\leq\\
  & \!\!\sum_{r_1,...,r_{2K}=0,1}\!\!|m_{r_1,...,r_{2K}}|\abs{\tr{\rho c^{r_1}_{1}...c^{r_{2K}}_{2K}}-\tr{\rho \av{c^{r_1}_{1}...c^{r_{2K}}_{2K}}}}.
  \end{split}
 \end{equation}
 We choose the $c_{\alpha}$ such that the covariance matrix, defined by
 \begin{equation}
 \Gamma_{ij}=\frac{i}{2}\tr{\rho[c_{i},c_{j}]},
 \end{equation}
 is in block diagonal form \cite{Greplova13},
 \begin{equation}
  \Gamma=\bigoplus_{n=1}^{K}\left(\begin{matrix} 0 & \lambda_n \\ -\lambda_n & 0 \end{matrix}\right),
 \end{equation}
with $\lambda_n=\frac{i}{2}\tr{\rho[c_{2n-1},c_{2n}]}$.  Let $R$ be the bit string $(r_1,...,r_{2K})$, and define $2w=\sum_i r_i$.  Because the state is Gaussian, we have \cite{Greplova13}
\begin{equation}
 \tr{\rho c^{r_1}_{1}...c^{r_{2K}}_{2K}}=i^w\mathrm{pf}\left(\Gamma^R\right),
\end{equation}
where $\mathrm{pf}$ is the Pfaffian, and $\Gamma^R$ is a submatrix of $\Gamma$ restricted to the rows and columns labeled by $i$ corresponding to $r_i=1$.  The Pfaffian is zero if any row or column is zero, so $\mathrm{pf}\left(\Gamma^R\right)=0$, unless the string $R$ only contains consecutive pairs of ones and zeros, such as $(1,1,0,0,1,1,...)$.

The Pfaffian has two useful properties.  The first is that $\mathrm{pf}\left(\begin{smallmatrix} 0 & \lambda \\ -\lambda & 0 \end{smallmatrix}\right)=\lambda$, and the second is that $\mathrm{pf}(A_1\oplus ... \oplus A_n)=\mathrm{pf}(A_1)...\mathrm{pf}(A_n)$.  For any string $R$ giving a non zero Pfaffian, we get
\begin{equation}
 \tr{\rho c^{r_1}_{1}...c^{r_{2K}}_{2K}}=i^w\prod_{n=1}^{K}\lambda_n^{r_{2n}}.
\end{equation}
And so
\begin{equation}
\begin{split}
  &\abs{\tr{\rho c^{r_1}_{1}...c^{r_{2K}}_{2K}}-\tr{\rho\av{c^{r_1}_{1}...c^{r_{2K}}_{2K}}}}\\
  &\leq\abs{\prod_{n=1}^{K}\lambda_n^{r_n}-\av{\prod_{n=1}^{K}\lambda_n^{r_n}}}\\
  &\leq \abs{\prod_{n=1}^{K}\lambda_n^{r_n}-\prod_{n=1}^{K}\av{\lambda_n^{r_n}}}+\avg{\abs{\prod_{n=1}^{K}\lambda_n^{r_n}-\prod_{n=1}^{K}\av{\lambda_n^{r_n}}}}\\
  &\leq  \sum_{n=1}^{K}\abs{\lambda_n^{r_n}-\av{\lambda_n^{r_n}}}+\sum_{n=1}^{K}\avg{\abs{\lambda_n^{r_n}-\av{\lambda_n^{r_n}}}},
 \end{split}
\end{equation}
where the third line follows from the triangle inequality and $|\av{f(t)}|\leq \av{|f(t)|}$.  The last line follows from the triangle inequality, and repeated applications of $|xy-\av{x}\av{y}|\leq |x-\av{x}|+|y-\av{y}|$, which uses $|x|$, $|y|\leq 1$.  So we need to focus on
\begin{equation}
 \abs{\lambda_n-\av{\lambda_n}}=\abs{\frac{i}{2}\tr{\rho[c_{i},c_{j}]}-\frac{i}{2}\tr{\rho\av{[c_{i},c_{j}]}}},
\end{equation}
where $i=2n-1$ and $j=2n$.

For each pair $c_i$ and $c_j$, we can always define creation operators $d_i^{\dagger}$ and $d_j^{\dagger}$, with 
\begin{equation}
\begin{split}
 c_i & =i(d_i-d_i^{\dagger})\\
 c_j & =d_j+d_j^{\dagger}.
 \end{split}
\end{equation}
We are not assuming that these creation operators correspond to orthogonal modes.  Using $\tr{\rho d_id_j}=0$, which follows because $[\rho,N]=0$, we have
\begin{equation}
\begin{split}
 \frac{i}{2}\tr{\rho [c_{i},c_j]} & =\tr{\rho (d_i^{\dagger}d_j-d_id_j^{\dagger})}\\
 & =\tr{\rho (d_i^{\dagger}d_j+d_j^{\dagger}d_i)}-\kappa,
 \end{split}
\end{equation}
where $\kappa=\{d_i,d_j^{\dagger}\}$ is a constant.
We can rewrite this as
\begin{equation}
 \frac{i}{2}\tr{\rho [c_{i},c_j]}=\tr{\rho b_i^{\dagger}b_i}-\tr{\rho b_j^{\dagger}b_j}-\kappa,
\end{equation}
where $b_i=\f{1}{\sqrt{2}}(d_i+d_j)$ and $b_j=\f{1}{\sqrt{2}}(d_i-d_j)$.
This allows us to apply corollary \ref{th:76} to get
\begin{equation}
\begin{split}
 \avg{\abs{\frac{i}{2}\tr{\rho[c_{i},c_{j}]}-\frac{i}{2}\tr{\rho\av{[c_{i},c_{j}]}}}}_T &\\
 \leq 2l\sqrt{c_3\left[\frac{D_G}{d}+ c_2\frac{n_{\max}}{T}\right]}&,
\end{split}
\end{equation}
where $l$ is the number of sites the measurement acts on.  Putting this all together, we get
\begin{equation}
 \begin{split}
  & \abs{\tr{\rho M(t)}-\tr{\rho \av{M}}}\\
  & \leq 4m^{\prime}lK\sqrt{c_3\left[\frac{D_G}{d}+ c_2\frac{n_{\max}}{T}\right]},
  \end{split}
 \end{equation}
 where
 \begin{equation}
  m^{\prime}= \!\!\sum_{r_2,r_4,...,r_{2K}=0,1}\!\!|m_{r_2,r_2,r_4,r_4,...,r_{2K}}|,
 \end{equation}
 where the sum only counts terms labelled by consecutive pairs of ones and zeros.  Then, for simplicity, we can use $m^{\prime}\leq m2^K$, where $m=\max |m_{r_1,...,r_{2K}}|$, since there are $2^K$ terms that contribute.  It may also simplify things to use $K\leq ls$.  This gives
 \begin{equation}
 \begin{split}
  & \abs{\tr{\rho M(t)}-\tr{\rho \av{M}}}\\
  & \leq 2^{ls+2}msl^2\sqrt{c_3\left[\frac{D_G}{d}+ c_2\frac{n_{\max}}{T}\right]}.
  \end{split}
 \end{equation}
 \end{proof}

}

\end{document}